\documentclass[11pt]{article}
 
\usepackage{fullpage,appendix}
\usepackage[bookmarks,colorlinks,breaklinks]{hyperref}  
\hypersetup{linkcolor=blue,citecolor=blue,filecolor=blue,urlcolor=blue} 
\usepackage{amsmath,amsfonts,amsthm,amssymb,xspace,bm}

\usepackage{amsmath,amsfonts,amssymb,xspace,bm}

\newcommand{\newref}[2][]{\hyperref[#2]{#1~\ref*{#2}}}
\numberwithin{equation}{section}

\newcommand{\sref}[1]{\newref[Section]{#1}}
\newcommand{\dref}[1]{\newref[Definition]{#1}}
\newcommand{\tref}[1]{\newref[Theorem]{#1}}
\newcommand{\lref}[1]{\newref[Lemma]{#1}}

\newcommand{\cref}[1]{\newref[Corollary]{#1}}

\newcommand{\eref}[1]{\newref[Equation]{#1}}

\newcommand{\clref}[1]{\newref[Claim]{#1}}
\newcommand{\gsp}{G_{sp}}
\newcommand{\spn}{S^{n-1}}
\newcommand{\gnn}{G_{\mathcal{N}}}

\newtheorem{theorem}{Theorem}[section]
\newtheorem{lemma}[theorem]{Lemma}
\newtheorem{claim}[theorem]{Claim}

\newtheorem{corollary}[theorem]{Corollary}

\newtheorem{definition}[theorem]{Definition}
\newtheorem{fact}[theorem]{Fact}

\renewcommand{\epsilon}{\varepsilon}

\DeclareMathOperator*{\pr}{\mathsf{Pr}} 

\DeclareMathOperator*{\ex}{\mathbb{E}}
\newcommand{\rgta}{\rightarrow}
\newcommand{\lfta}{\leftarrow}
\newcommand{\iprod}[2]{\langle #1,#2\rangle}

\newcommand{\rn}{\mathbb{R}^n}
\newcommand{\rnn}{\mathbb{R}^{n \times n}}

\newcommand{\reals}{\mathbb{R}}

\newcommand{\dpm}{\{1,-1\}}
\newcommand{\oo}{\dpm}
\newcommand{\dpmn}{\dpm^n}
\newcommand{\sign}{\mathsf{sign}}

\newcommand{\hh}{\mathcal{H}}

\newcommand{\gsn}{\mathcal{N}(0,1)^n}
\newcommand{\calC}{{\mathcal C}}
\newcommand{\ugs}{\mathcal{N}}

\newcommand{\NS}{\mathbb{NS}}

\newcommand{\eps}{\epsilon}

\newcommand{\note}[1]{\marginpar{\tiny *note in TeX*}}
\newcommand{\ignore}[1]{}

\newcommand{\calD}{{\cal D}}
\newcommand{\calN}{{\cal N}}

\renewcommand{\phi}{\varphi}

 \newcommand{\R}{\mathbb{R}}

 \newcommand{\NN}{{{\cal N}^n}}
\newcommand{\E}{E}

 \newcommand{\poly}{\mathrm{poly}}

 \newcommand{\eqdef}{\stackrel{\textrm{def}}{=}}

 \newcommand{\C}{{{\cal C}}}
 \newcommand{\D}{{{\cal D}}}
 
 \newcommand{\ptp}{\mathcal{K}}

 \newcommand{\fak}{f^{\wedge k}}
 \newcommand{\psil}{\|\psi^{(4)}\|_1}
 
 \newcommand{\rnk}{\reals^{n\times k}}
 
 \newcommand{\greg}{G}
 \newcommand{\zo}{\{0,1\}}

\title{An Invariance Principle for Polytopes}
\author{PRAHLADH HARSHA\\Tata Institute of Fundamental Research, Mumbai\and
ADAM KLIVANS\\University of Texas, Austin \and 
RAGHU MEKA\\University of Texas, Austin}
\begin{document}
\maketitle
  \begin{abstract}
    Let $X$ be randomly chosen from $\{-1,1\}^n$, and let $Y$ be randomly
   chosen from the standard spherical Gaussian on $\R^n$.  For any (possibly unbounded) polytope $P$
   formed by the intersection of $k$ halfspaces, we prove that 
 $$\left|\Pr\left[X  \in P\right] - \Pr\left[Y \in P\right]\right| \leq \log^{8/5}k \cdot \Delta,$$ where $\Delta$ is
 a parameter that is small for polytopes formed by the intersection of
 ``regular'' halfspaces (i.e., halfspaces with low influence). The
 novelty of our invariance principle is the polylogarithmic dependence
 on $k$. Previously, only bounds that were at least linear in $k$ were
 known.  The proof of the invariance principle is based on a
 generalization of the Lindeberg method for proving central limit
 theorems and could be of use elsewhere.
 
 We give two important applications of our invariance principle, one from learning theory and the other from pseudorandomness:
\begin{enumerate}
\item A bound of $\log^{O(1)}k \cdot {\epsilon}^{1/6}$ on the Boolean
  noise sensitivity of intersections of $k$ ``regular'' halfspaces
  (previous work gave bounds linear in $k$). This gives a
  corresponding agnostic learning algorithm for intersections of
  regular halfspaces.
 
\item A pseudorandom generator (PRG) for estimating the Gaussian
  volume of polytopes with $k$ faces within error $\delta$ and
  seed-length $O(\log n\,\poly(\log
  k,1/\delta))$. 
 \end{enumerate}
 We also obtain PRGs with similar parameters that fool polytopes
 formed by intersection of regular halfspaces over the
 hypercube. Using our PRG constructions, we obtain the first
 deterministic quasi-polynomial time algorithms for approximately
 counting the number of solutions to a broad class of integer
 programs, including dense covering problems and contingency tables.
 \end{abstract}

\section{Introduction: Invariance Principles in Theoretical Computer Science}
An important theme in theoretical computer science over the last two
decades has been the usefulness of translating a combinatorial problem
over a discrete domain (e.g., $\{-1,1\}^n$) to a problem in continuous
space.  The notion of convex relaxation, for example, is now a
standard technique in the design of algorithms for optimization
problems.  More recently, the study of analytic properties of Boolean
functions (e.g., Fourier spectra and sensitivity) has been a
fundamental tool for proving results in hardness of approximation
\cite{deWolf2008,ODonnell2008} and learning theory \cite{Mansour1994}.

The influential work of Mossel, O'Donnell, and Oleszkiewicz
\cite{MosselOO2005} proving the ``Majority Is Stablest''
conjecture has led to a rich collection of hardness results for
constraint satisfaction problems, most notably for the Max-Cut
problem.  The crux of their work is an invariance principle relating
the behavior of low-degree polynomials over the uniform measure on
$\{-1,1\}^n$ to their behavior with respect to Gaussians:

 \begin{theorem}[invariance
   principle for polynomials~\cite{MosselOO2005}]\footnote{Similar invariance
     principles were also shown by Chatterjee \cite{Chatterjee2005}
     and Rotar \cite{Rotar1979}.}
 Let $P$ be a multilinear polynomial such that $\|P\| = 1$. Then, for any $t\in \reals$,

 \[ \left| \Pr_{x \in_u \{-1,1\}^n}\left[P(x) > t\right] - \Pr_{x \lfta \calN^n}\left[P(x) > t\right] \right| \leq \tau .\]
 \end{theorem}

 Here $\calN^n$ is the standard multivariate spherical Gaussian distribution on $\R^n$; the
 parameter $\tau$ depends on the coefficients of $P$ and is small if $P$ is ``regular'' in the sense that the ``influence'' of each variable in $P$ is small.

 Roughly speaking, the above invariance principle says that the cumulative distribution function (cdf) of a
 polynomial over $\{-1,1\}^n$ is close to the cdf of a polynomial
 over $\calN^n$ if the coefficients of the polynomial are sufficiently
 regular.  

 Additionally, the above invariance principle and its generalizations in \cite{Mossel2008} have had a wealth of powerful applications in the following areas: hardness of approximation (see \cite{DinurMR2009,Austrin2007,AustrinM09,Raghavendra2008,ODonnellW2009,BansalK2010}, hardness of learning (see \cite{FeldmanGRW2009}), social choice theory (see \cite{Mossel2011}), testing (see \cite{BlaisO2010}), graph products (see \cite{DinurFR2008}) and the analysis of Boolean functions (see \cite{DiakonikolasHKMRST2010}) among others.  Invariance principles are now widely considered to be powerful tools in computational complexity theory.  As such, it is important to continue to understand, quantitatively, how a function's cumulative density function changes when translating from one underlying distribution to another.




\subsection{An Invariance Principle for Polytopes}
 The main result of this paper is an invariance principle for
 characteristic functions of polytopes.  Recall that a polytope $\ptp$ is
 a (possibly unbounded) convex set in $\R^n$ formed by the intersection
 of some finite number of supporting halfspaces.  We refer to $\ptp$ as a
 $k$-polytope if it is equal to the intersection of $k$ halfspaces.  To state our result we need the notion of {\it regularity}.

\begin{definition}[regularity]\label{def:regularity}
A vector $u \in \reals^n$ is $\epsilon$-regular if $\sum_i u_i^4 \leq \epsilon^2 \|u\|_2^2$. A matrix $W \in \reals^{n \times k}$ is $\epsilon$-regular if every column of $W$ is $\epsilon$-regular. A polytope $\ptp = \{x: W^Tx \leq \theta\}$ is $\epsilon$-regular if $W$ is $\epsilon$-regular~\footnote{``Regular polytopes'' have a different meaning in combinatorics, but for the purpose of this paper, we will abuse notation and say a polytope is $\epsilon$-regular if it is formed by the intersection of $\epsilon$-regular halfspaces as in \dref{def:regularity}.}. 
\end{definition}

We will require our polytopes to be sufficiently regular to apply our invariance principle.  This is necessary even in the case of a single halfspace; the function $1x_{1} + 0x_{2} + \ldots + 0x_{n} = x_{1}$ where each $x_{i} \in \{-1,1\}$ will never converge to a Gaussian (this linear function is highly non-regular). Note that regularity does not depend on the threshold vector $\theta$.

Our main theorem is as follows (see \tref{th:ipmain} for exact statement):

 \begin{theorem}[invariance principle for polytopes]\label{thm:main}
For $\ptp$ a $\epsilon$-regular $k$-polytope,
 \[ \left| \Pr_{x \in_u \{-1,1\}^n} \left[x \in \ptp\right] - \Pr_{x \lfta \calN^n}\left[x \in \ptp\right] \right| \leq C \log^{8/5}k\cdot \epsilon^{1/6} .\]

 \end{theorem}

Our invariance principle also holds more generally for any product distribution that is hypercontractive and whose first four moments are appropriately bounded (often in this paper we focus on the special case of uniform on $\{-1,1\}^n$). 

 The novelty of our theorem is the dependence of the error on $k$.
 Applying a recent result due to Mossel \cite{Mossel2008}, it is possible to obtain a statement similar to \tref{thm:main} with an error term that has a polynomial dependence on $k$.  Achieving polylogarithmic dependence on $k$, however, is much harder, and we need to use some nontrivial results from the analysis of convex sets in Gaussian space. We remark that our result is optimal up to polylogarithmic factors: any invariance principle as above cannot have an error bound of $o(\epsilon \cdot \sqrt{\log k})$ (see \sref{sec:lbound}).

 The case $k=1$, a single halfspace, is equivalent to the classical Berry-Ess\'een theorem \cite{Feller1}, a fundamental theorem from probability and statistics giving a quantitative version of the Central Limit Theorem.  We can therefore view our principle as a generalization of the Berry-Ess\'een theorem for polytopes.

\subsection{Applications of Our Invariance Principle}

While we believe the statement of our main theorem is interesting in and of itself, we apply our invariance principle to obtain striking new results in various subfields of theoretical computer science:

\begin{itemize}

\item The Analysis of Boolean Functions:  we give new bounds on the {\em Noise Sensitivity} of characteristic functions of polytopes.

\item Learning Theory:  we give the best known algorithms for (agnostically) learning intersections of halfspaces with respect to the uniform distribution on $\{-1,1\}^n$.

\item Pseudorandomness: we build pseudorandom generators for polytopes and give the first deterministic algorithms for approximately counting the number of solutions to broad classes of integer programs.

\end{itemize}

We elaborate on these applications below.  More generally, our main theorem gives new insight on the structure of integer points in polytopes (that is, solutions to integer programs).  Understanding this structure is an important topic in computer science \cite{BarvinokV2008}, optimization \cite{Ziegler-polytope}, and combinatorics \cite{BeckRobins}, and we believe our invariance principle will find many future applications. 

\subsection{Application: Bounding the Noise Sensitivity of Intersections of Halfspaces}

The noise sensitivity of Boolean functions, introduced in the seminal
works of Kahn, Kalai and Linial \cite{KahnKL1988} and Benjamini, Kalai and Schramm \cite{BenjaminiKS1999}, is an important notion in the {\sl analysis} of Boolean functions. Roughly speaking, the noise sensitivity of a Boolean function $f$ measures the probability over a randomly chosen input $x$ that $f$ changes sign if each bit of $x$ is flipped independently with probability $\delta$.

Bounds on the noise sensitivity of Boolean functions have direct applications in hardness of approximation \cite{Hastad2001,KhotKMO2007}, hardness amplification \cite{ODonnell2004}, circuit complexity \cite{LinialMN1993}, the theory of social choice \cite{Kalai2005}, and quantum complexity \cite{Shi2000}. Here, we focus on applications in learning theory, where it is known that bounds on the noise sensitivity of a class of Boolean functions yield learning algorithms that succeed in harsh noise models such as the agnostic model of learning \cite{KalaiKMS2008}.  

A direct application of our invariance principle \tref{thm:main} gives the following new bound on the noise sensitivity of intersections of regular halfspaces:

 \begin{theorem}[noise sensitivity of intersections of halfspaces] \label{thm:mainlearn}
 Let $f$ be computed by the intersection of $k$, $\eps$-regular halfspaces.  Then the Boolean noise sensitivity of $f$ for noise rate $\epsilon$ is at most $(\log k)^{O(1)} \cdot \epsilon^{1/6}$.
 \end{theorem}

The current best bound for the noise sensitivity of intersection of $k$ arbitrary halfspaces is $O(k\sqrt{\epsilon})$. This bound is obtained by starting with the $\sqrt{\epsilon}$ noise sensitivity bound for a single halfspace due to Peres \cite{Peres2004} and applying a union bound over $k$ halfspaces. On the other hand, optimal bounds of $\Theta(\sqrt{\log k} \sqrt{\epsilon})$ for the related Gaussian noise sensitivity were obtained recently by Klivans, O'Donnell and Servedio \cite{KlivansOS2008}.  Our result is an important step towards improving noise sensitivity bounds for intersections of arbitrary (not necessarily regular) halfspaces.  We believe that the right order for Boolean noise sensitivity of intersection of $k$ halfspaces is $\Theta(\sqrt{\log k} \sqrt{\epsilon})$ as well. 

\subsection{Application: Learning Intersections of Halfspaces}

We give new result for agnostically learning intersections of halfspaces with respect to the uniform distribution on $\{-1,1\}^n$.  Learning intersections of halfspaces (i.e., convex sets) is a fundamental challenge from learning theory.  Distribution-free learning of even an intersection of two halfspaces remains a challenging open problem.  A natural restriction of the problem is to assume the underlying distribution is uniform over $\{-1,1\}^n$ (this can be seen to be more difficult than the case where the underlying distribution is Gaussian). 

Applying a result of Kalai et al.~\cite{KalaiKMS2008} and Klivans et al.~\cite{KlivansOS2004},  \tref{thm:mainlearn} implies the following:

\begin{theorem}[learning intersections of halfspaces] \label{thm:mainlearnstate}
The conept class of intersections of $k$ halfspaces are agnostically  learnable with respect to the uniform distribution on $\{-1,1\}^n$ in time $n^{(\log^{O(1)} k)}$ for any constant error parameter.
\end{theorem}

Agnostic learning corresponds to learning with adversarial noise (see \sref{sec:ag} for a precise definition). In particular, intersections of $\{-1,1\}$ halfspaces (oriented majorities) are $\eps$-regular and fall into this class. The previous best algorithm for learning these concept classes, even in the easier PAC model, ran in time $n^{O(k^2)}$ (\cite{KlivansOS2004,KalaiKMS2008}).

The obvious remaining open problem here is to agnostically learn intersections of arbitrary (not necessarily regular) halfspaces with respect to the uniform distribution while preserving the quasipolynomial-time dependence on the number of halfspaces.  Typically, handling the regular case is the first step towards such a result, and we have accomplished that here for the first time.

 \subsection{Application: Pseudorandomness for Polytopes} \label{sec:prgthms}

 Our invariance principle also yields new results for several problems in derandomization.  In particular, we give the first deterministic algorithms for approximately counting the number of solutions to broad classes of integer programs.  Recall the following definition of pseudorandom generators (PRGs):
 \begin{definition}
 Let $\mu$ be a distribution over $\reals$. A function $G:\zo^r \rgta \dpm^n$ is said to $\delta$-fool a polytope $\ptp$ with respect to $\mu$ if the following holds.
 \[ \left|\pr_{y \in_u \zo^r}\left[G(y) \in \ptp\right] - \pr_{X \lfta \mu^n}\left[X \in \ptp\right]\right| \leq \delta.\]
 \end{definition}

Combining our invariance principle with a PRG similar to a recent construction of Meka and Zuckerman \cite{MekaZ2010}, we obtain the following pseudorandom generator:


 \begin{theorem}[PRGs for regular polytopes]\label{th:prgmain}
 For all $\delta \in(0,1)$, there exists an explicit PRG $G:\zo^r \rgta \dpm^n$ with $r = O((\log n \log k)/\epsilon)$ that $\delta$-fools all polytopes formed by the intersection of $k$ $\epsilon$-regular halfspaces with respect to all proper and hypercontractive distributions $\mu$ for $\epsilon = \delta^5/(\log^{8.1} k)(\log(1/\delta))$.
 \end{theorem}

 The constants above depend on the hypercontractivity constants of $\mu$.  We define proper and hypercontractive
 distributions in the next section and remark that the uniform distribution over $\{-1,1\}^n$ and the Gaussian
 distribution are examples of such distributions.

This pseudorandom generator gives an algorithm for approximately counting the number of $\{-1,1\}^n$ points in polytopes formed by the intersection of regular halfspaces.  Put another way, given an integer program whose constraints are sufficiently regular, we give a quasi-polynomial time, deterministic algorithm for approximately counting the number of $\{-1,1\}^n$ solutions:

\begin{corollary}[Approximate counting for regular integer programs] \label{cor:approx}
Let $A$ be an integer program with $n$ variables and $k$ constraints where each constraint is an $\epsilon$-regular halfspace (we define regularity precisely in \sref{sec:notation}).  For $\epsilon = \delta^{5}/(\log^{8.1} k)(\log(1/\delta))$ there exists a deterministic algorithm that runs in time $\mathsf{exp}(O(\log n \log k)/\epsilon)$ for estimating the number of $\{-1,1\}^n$ points satisfying $A$ to within an additive $\epsilon 2^{n}$.
\end{corollary}


\cref{cor:approx} implies quasi-polynomial time, deterministic,
approximate counting algorithms for a broad class of integer programs.
For example, dense covering programs such as dense set-cover, and
$\{0,1\}$-contingency tables correspond to polytopes formed by the
intersection of $\eps$-regular halfspaces.  For these types of integer
programs, we can deterministically approximate, to within an additive
error $\epsilon$, the fraction of the hypercube that are integer
solutions in quasi-polynomial time (i.e., we obtain an {\em additive}
approximation of the number of integer solutions to within $\epsilon 2^n$).

While there has been much work on approximately counting solutions to
integer programs using randomized algorithms, we are unaware of
results giving {\em deterministic} algorithms for these tasks (even
for the case of regular integer programs) that run in subexponential
time in the number of constraints.  Very recently, there has been work
on deterministic approximate counting for the multidimensional
knapsack problem and the contingency table
problem~\cite{GopalanKM2010}, but these algorithms still run in time
exponential in the number of constraints.  Another difference between
these results and the algorithm of Gopalan et al.~\cite{GopalanKM2010} is that
\cite{GopalanKM2010} give (stronger) relative-error guarantees, while here we give additive approximations.  For a further discussion of related work see \sref{sec:relatedwork} and \sref{sec:prgreg}.

\subsection{Additional Invariance Principles}
 As stated, our invariance principle applies to polytopes whose bounding hyperplanes have coefficients that are sufficiently regular.  In some cases, however, we can randomly rotate an arbitrary polytope so that all the bounding hyperplanes become regular.   As such, after applying a suitable random transformation (which we derandomize), we can build PRGs for {\em arbitrary} polytopes if the underlying distribution is spherically symmetric (e.g., Gaussian):

 \begin{theorem}[PRGs for Polytopes in Gaussian space] \label{th:prgnormal}
 For a universal constant $c > 0$ and all $\delta > c \log^2 k/n^{1/11}$, there exists an explicit PRG $\gnn:\zo^r \rgta \reals^n$  with $r = O((\log n)(\log^{9.1}k)/\delta^{5.1})$ that $\delta$-fools all $k$-polytopes with respect to $\mathcal{N}$.
 \end{theorem}


 Additionally, we prove an invariance principle for polytopes with respect to the uniform distribution over the $n$-dimensional sphere $S^{n-1}$.  This allows us to easily modify our PRG for polytopes in Gaussian space and build PRGs for intersections of spherical caps:

 \begin{theorem}[PRGs for intersections of spherical caps]\label{th:prgspherical}
 For a universal constant $c > 0$ and all $\delta > c \log^2 k/n^{1/11}$, there exists an explicit PRG $\gsp:\zo^r \rgta \spn$ with $r =  O((\log n) (\log^{9.1}k)/\delta^{5.1})$ that $\delta$-fools all $k$-polytopes with respect to the uniform distribution over $S^{n-1}$.
 \end{theorem}

An immediate consequence of the above PRG construction is a polynomial time derandomization of the Goemans-Williamson
approximation algorithm for Max-Cut~\cite{GoemansW1995} and other similar hyperplane based randomized rounding
schemes. Observe that this derandomization is a {\em black-box} derandomization as opposed to some earlier
derandomizations of the Goemans-Williamson algorithm, which are instance-specific (e.g., \cite{MahajanH1999}).

\subsection{Proof Outline of the Main Theorem}

In this section, we give a high level outline of the proof of our
invariance principle and contrast it with the techniques of Mossel et al.~\cite{MosselOO2005} and Mossel \cite{Mossel2008}.  The proof proceeds in two steps.  \\ 

\noindent {\bf Step One:} As in \cite{MosselOO2005} and
\cite{Mossel2008}, we first use the Lindeberg (or ``replacement'')
method\footnote{The Lindeberg method entails replacing each of the
  $X_{i}$'s with $Y_{i}$'s one step at a time and bounding the error
  in each step. This is more commonly referred to as the hybrid
  argument in theoretical computer science literature since the
  intermediate random variables are a hybrid of both $X$ and $Y$.} (see \cite{PR-clt}) to prove an invariance principle for smooth functions.  By this we mean proving that
\begin{equation}\label{eq:ipsmoothtoy}
 \left| \ex_{X \in \{-1,1\}^n}\left[\Psi(\ell_{1}(X),\dots,\ell_{k}(X))\right] - \ex_{Y \in \calN^n}\left[\Psi(\ell_{1}(Y),\dots,\ell_{k}(Y))\right] \right| \leq \gamma,  
\end{equation}
where $\ell_{1},\ldots,\ell_{k}$ are linear functions (corresponding to the normals of the faces of the $k$-polytope) and $\Psi$ is a smooth function.  The value $\gamma$ will depend on $k$,  the coefficients of the $\ell_p$'s and the derivatives of $\Psi$.  The function $\Psi$ is often called a ``test'' function and is smooth if there is a uniform bound on its fourth derivative.  Notice here that $\Psi$ maps $\R^{k}$ to $\R$; in \cite{MosselOO2005}, they were concerned with the value $\Psi(Q(X))$ for a low-degree polynomial $Q$ and a univariate test function $\Psi$.

At this point, we could take $\Psi$ to be the $k$-wise product of a test function constructed by Mossel et al.~to
approximate the logical AND function. Further, Mossel provides a very general framework for obtaining multivariate test functions and gives bounds for the overall error incurred by the hybrid argument.  Here we run into our first difficulty:  the standard hybrid argument as used by Mossel et al.~and Mossel results in a bad dependence on the coefficients of the $\ell_p$'s. In particular, the resulting error term is not small even for polytopes formed by the intersection of regular halfspaces. 

To solve this problem, we use a non-standard hybrid argument that groups the input variables into blocks.  We observe
that in the Lindeberg method it is irrelevant in which order we replace $X_{i}$'s with $Y_{i}$'s -- in fact a random order would suffice.   Further, we can group the $X_{i}$'s into blocks and proceed blockwise with the hybrid argument. To implement this intuition, we partition $[n]$ randomly into a set of blocks and replace all the $X_{i}$'s within a block by the corresponding $Y_{i}$'s one block at a time. Proceeding in this fashion with a random partitioning has a ``smoothing effect'' on the coefficients of the linear functions resulting in a much better bound on the error in terms of the coefficients.

Roughly speaking, if $\ell_{pi}$ denotes the $i$'th coefficient of $\ell_p$, then the standard hybrid arguments of \cite{PR-clt}, \cite{MosselOO2005}, \cite{Mossel2008} incur an error proportional to $\sum_{i \in [n]} \left(\max_{p \in [k]} |\ell_{pi}|^{4}\right)$, which can be as large as $\Omega(k)$ even for regular functions $\ell_p$. In contrast, our {\sl randomized-blockwise-hybrid} argument only suffers an error of  $(\log k)\cdot \max_{p \in [k]} \sum_i |\ell_{pi}|^4$, which is small for regular functions. It turns out that in the above analysis, we can choose the random partitioning into blocks in a $\Theta(\log k)$-wise independent manner, instead of uniformly at random, and this is crucial for our PRG constructions.\\

\noindent{\bf Step Two:}  Given the above invariance principle for smooth functions, we now aim to translate the closeness in expectation for smooth functions to closeness in cdf distance.  Here the smoothness of the test function $\Psi$ becomes important, and we run into our second problem: the natural choice of test function $\Psi$ (the multivariate version of the test function from Mossel et al.) leads to an error bound on the order of $k$, rather than $\poly(\log k)$.  To get around this problem,  we first observe that in Mossel's proof of the multivariate invariance principle as in our {\sl randomized-blockwise-hybrid} argument, it suffices to bound the `$l_1$-norm' of the fourth derivative $\sup_{x \in \reals^k}(\sum_{p,q,r,s\in[k]} |\partial_p\partial_q\partial_r\partial_s \Psi (x)|)$, instead of uniformly bounding the fourth derivative $\sup_{x \in \reals^k,p,q,r,s \in [k]}(|\partial_p\partial_q\partial_r\partial_s \Psi (x)|)$. Thus, it suffices to obtain a smooth approximation of the AND function for which the former quantity is small. Fortunately for us, we have uncovered a beautiful result due to Bentkus \cite{Bentkus1990}, who constructs a smooth approximation of the AND function with precisely this property. 

The final difficulty for translating closeness in expectation as in
\eref{eq:ipsmoothtoy} to closeness in cdf distance is to prove that
$\Psi$ differs from the characteristic function only on a set of small
Gaussian measure.  To this end, we show that it suffices to bound the
Gaussian measure of $l_\infty$-neighborhoods around the boundary of
$k$-polytopes.  For an $l_\infty$-neighborhood of width $\lambda$, a
union bound would imply Gaussian measure on the order of $k \lambda$.
At this point, however, we can apply a 
result due to Nazarov \cite{Nazarov2003} on the Gaussian surface area of $k$-polytopes to get the much better bound of $\sqrt{\log k}\,\lambda$.  This result of Nazarov was used before by Klivans, O'Donnell and Servedio \cite{KlivansOS2008} in the context of learning intersections of halfspaces with respect to Gaussian distributions.\\

We give an outline of the proofs of the applications of the invariance principle to noise sensitivity and PRGs in the corresponding sections.

 \subsection{Related Work} \label{sec:relatedwork}
\ignore{Central limit theorems or invariance principles are a long-studied topic in the probability community with a variety of
quantitative formulations known under various settings. The most prominent among these is the}
 As mentioned earlier, the classical Berry-Ess\'een theorem \cite{Feller2} from probability, a quantitative version of the Central Limit Theorem, gives an invariance principle for the case of a single halfspace (i.e., $k=1$).  More precisely, for any $w \in \reals^n$, such that $\|w\| = 1$ and each coefficient of $w$ is at most $\epsilon$, the Berry-Ess\'een theorem states that 
\[ \left|\Pr_{x \in \{-1,1\}^n} \left[\iprod{w}{x} \geq t\right] - \Pr_{x \lfta \calN^{n}}\left[\iprod{w}{x} \geq t\right]\right| \leq O(\epsilon).\]  

 Bentkus \cite{Bentkus2003} proves a multidimensional Berry-Ess\'een theorem for sums of vector-valued random variables each with identity covariance matrix, whose error term depends on the Gaussian surface area of the test set. Although his paper deals with topics related to our work, his result seems to have no implications in our setting.


 There is a long history of research on approximately counting the number of solutions to integer programs, especially with regard to contingency tables \cite{JerrumS1997,CryanD2003}. However, not much is known in terms of {\em deterministic} algorithms, and we believe that our deterministic quasi-polynomial time algorithms for dense covering problems and dense set cover instances is the first result of its kind. 

Regarding contingency tables, Dyer \cite{Dyer2003} gave a randomized
relative-error approximation algorithm for counting solutions to
contingency tables that runs in time exponential in the number of
rows. In contrast, we obtain an algorithm that runs in
quasi-polynomial time in the number of rows (however, we do not give a
relative-error approximation). Although not stated explicitly before,
it is easy to see that the pseudorandom generator for small space
machines of Impagliazzo, Nisan and Wigderson \cite{ImpagliazzoNW1994} yields a deterministic algorithm
for counting $n\times k$ contingency tables with additive error at
most $\epsilon 2^n$ and run time $2^{O(\log^2(nk/\epsilon))}$. This is incomparable to our algorithm for contingency tables which has run time $2^{(\log n) \cdot \poly(\log k,1/\epsilon)}$.  In our case, we obtain a polynomial-time, black-box derandomization for contingency tables with a constant number of rows (for $\epsilon = O(1)$).


 For PRGs for intersections of halfspaces, recently Gopala et al.~\cite{GopalanOWZ2010} and Diakonikolas, Kane and Nelson \cite{DiakonikolasKN2010} gave results incomparable to ours. Gopalan et al.~give generators for arbitrary intersections of $k$ halfspaces with seed length linear in $k$ but logarithmic in $1/\delta$. Diakonikolas et al.~show that bounded independence fools intersections of quadratic threshold functions and in particular, get generators with seed length $O((\log n)\cdot\poly(k,1/\epsilon))$ fooling intersections of $k$ halfspaces. Due to the at least linear dependence on $k$, the results of the above works do not yield good algorithms for counting solutions to integer programs, as in this setting $k$ is typically large (e.g., $\poly(n)$).


\subsection{Discussion and Future Work}


One obvious weakness of our applications to noise sensitivity bounds (\tref{thm:mainlearn}) and PRGs over the hypercube (\tref{th:prgmain}) is the regularity requirement. Recent results on sensitivity bounds and PRGs for halfspaces and PTFs (\cite{DiakonikolasHKMRST2010,MekaZ2010}) use certain {\sl regularity lemmas} which allow one to ``reduce'' the problem for arbitrary functions to the regular case and then use invariance to handle the regular case. Unfortunately, applying the reductions to the regular case as in the above works leads to bounds that are at least linear in $k$, even when using our stronger bounds for the regular case. We (optimistically) believe that the above difficulty could be overcome and a better reduction to the regular case can be achieved.

\newcommand{\rect}{\mathsf{Rect}}
\newcommand{\ok}{\mathsf{1}_k}

\section{Notation and Preliminaries} \label{sec:notation}
We use the following notation.
\begin{enumerate}
\item For $W \in \reals^{n \times k}$, $\theta \in \reals^k$, $\ptp(W,\theta)$ denotes the polytope $\ptp(W,\theta) = \{x: W^T x \leq \theta\}$. We say a polytope $\ptp(W,\theta)$ as above has $k$ faces. 
\item Unless stated otherwise, we work with the same polytope $\ptp(W,\theta)$ and assume that the columns of the matrix $W$ have norm one. We often shorten $\ptp(W,\theta)$ to $\ptp$ if $W,\theta$ are clear from context. We assume that $k \geq 2$. 
\item For $A \in \reals^{m_1 \times m_2}$, $A^T$ denotes the transpose of $A$ and for $p \in [m_2]$, $A^p$ denotes the $p$'th column of $A$. 

\item The all ones vector in $\reals^k$ is denoted by $\ok$. 
\item For $u \in \reals^k$, define rectangle 
\[ \rect(u) = (-\infty,u_1] \times (-\infty, u_2] \times \cdots \times (-\infty,u_k].\]
Note that $x \in \ptp(W,\theta)$ if and only if $W^T x \in \rect(\theta)$. 
\item $\NN$ (where $\calN = \calN(0,1)$) denotes the standard multivariate spherical Gaussian distribution over $\reals^n$ with mean $0$ and identity covariance matrix. 
\item For a 4-times differentiable function $\psi:\reals^k \rgta \reals$, let 
\[ \psil = \sup \, \left\{\, \sum_{p,q,r,s \in
  [k]}|\, \partial_p\partial_q\partial_r\partial_s\,\psi(a_1,\ldots,a_k)\,|\,:
\, (a_1,\ldots,a_k) \in \reals^k\,\right\}.\] 
We call $\psi$ a {\em smooth} function, if the above quantity is finite.
\item We denote all universal constants by $c,C$, even when we have in mind different constants in the same equation. Also, if left unspecified, we write $\|u\|$ for $\|u\|_2$. 
\end{enumerate}

\ignore{
\begin{definition}[regularity]\label{def:regularity}
A vector $u \in \reals^n$ is $\epsilon$-regular if $\sum_i u_i^4 \leq \epsilon^2 \|u\|^2$. A matrix $W \in \reals^{n \times k}$ is $\epsilon$-regular if every column of $W$ is $\epsilon$-regular. A polytope $\ptp = \ptp(W,\theta)$ is $\epsilon$-regular if $W$ is $\epsilon$-regular~\footnote{``Regular polytopes'' have a different meaning in combinatorics, but for the purpose of this paper, we will abuse notation and say a polytope is $\epsilon$-regular if it is formed by the intersection of $\epsilon$-regular halfspaces as in \dref{def:regularity}.}. 
\end{definition}

We will require our polytopes to be sufficiently regular to apply our invariance principle.  This is necessary even in the case of a single halfspace; note that the function $1x_{1} + 0x_{2} + \ldots + 0x_{n} = x_{1}$ where each $x_{i} \in \{-1,1\}$ will never converge to a Gaussian (this linear function is highly non-regular).}

The main results of this paper are applicable to a large class of product distributions that satisfy the following two properties.

\begin{definition}[proper distributions]
A distribution $\mu$ over $\reals$ is proper if for $X \lfta \mu$, $\ex[X] = 0$, $\ex[X^2] = 1$ and $\ex[X^3] = 0$.
\end{definition}
\begin{definition}[hypercontractive distributions]
 A distribution $\mu$ over $\reals$ is hypercontractive, if there exists a constant $c_\mu$ such that the following holds. For any $m$, vector $u \in \reals^m$, and any $q \geq 2$,
\[ \left(\ex_{X\lfta \mu^m}\left[\,|\iprod{u}{X}|^q\,\right]\right)^{1/q} \leq c_\mu \sqrt{q} \left(\ex_{X \lfta \mu^m}\left[\,|\iprod{u}{X}|^2\,\right]\right)^{1/2}.\] 
\end{definition}

Two important examples of product distributions that are proper and hypercontractive are the uniform distribution over the hypercube $\dpm^n$ and the multivariate spherical Gaussian $\NN$. 

We also use the following hypercontractivity inequality for degree $d$ multilinear polynomials over the hypercube (see \cite{janson1997} for instance).   

\begin{lemma}[$(2,q)$-hypercontractivity]\label{lm:hypercon}
For any $q \in [2,\infty)$ and any degree $d$ multilinear polynomial $P:\dpm^n \rgta \reals$, 
\[ \left(\ex_{x \in_u \dpm^n}\left[\,|P(x)|^q\,\right]\right)^{1/q} \leq q^{d/2}\,\left(\ex_{x \in_u \dpm^n}\left[\,|P(x)|^2\,\right]\right)^{1/2}.\]
\end{lemma}

We shall also use the following classical large-deviation inequality for Lipschitz functions in Gaussian space. For a function $f:\reals^n \rgta \reals$, the Lipschitz constant of $f$ is defined as $\|f\|_{Lip} = \sup\{|f(x) - f(y)|/\|x - y\|_2: x \neq y \in \reals^n\}$.
\begin{theorem}[\cite{LedouxT}]\label{th:gaussianld}
  For $f:\reals^n \rgta \reals$ with a bounded Lipschitz constant, $\mu(f) = \ex_{y \lfta \mathcal{N}^n}\left[f(y)\right]$, and $t > 0$,
\[ \pr_{x \lfta \mathcal{N}^n}\left[\, |f(x) - \mu(f)| > t\,\right] \leq 2 \exp(-t^2/2\|f\|_{Lip}).\]
\end{theorem}

\subsection{Agnostic Learning} \label{sec:ag}
Here we describe the agnostic framework of learning (a generalization of PAC learning) and describe how noise-sensitivity bounds translate into learning algorithms.  First we define noise sensitivity:

\begin{definition}[noise sensitivity]\label{def:noisesen}
Let $f$ be a Boolean function $f:\dpmn \to \dpm$.  For any $\delta \in (0,1)$, let $X$ be a random element of the hypercube $\dpmn$ and $Z$ a $\delta$-perturbation of $X$ defined as follows: for each $i$ independently, $Z_i$ is set to $X_i$ with probability $1-\delta$ and $-X_i$ with probability $\delta$. The noise sensitivity of $f$, denoted $\NS_\delta(f)$, for noise $\delta$ is then defined as follows: $\NS_\delta(f) = \Pr\left[f(X) \neq f(Z)\right]$.
\end{definition}

Now we describe the learning model of focus in this paper, {\em agnostic learning}.  In the agnostic learning framework \cite{KearnsSS1994,Haussler1992}, the learner receives labelled examples $(x,y)$ drawn from a fixed distribution over example-label pairs.

\begin{definition}[Agnostic Learning]
Let $\D$ be any distribution on $X \times \R$ and let $\C$ be a concept class of functions.  Define
\[\mathsf{opt}=\min_{f\in \C} \Pr_{(x,y)\sim \D}\left[f(x) \neq y\right].\] That is, $\mathsf{opt}$ is the
error of the best fitting concept in $\C$ with respect to $\D$.

We say that an algorithm $A$ agnostically learns a concept class $\C$ over
$\D$ if the following holds: for any $\D$ on $X \times \R$ with marginal distribution $D_{X}$ on $X$,
if $A$ is given random examples drawn from $\D$, then with high probability 
$A$ outputs a hypothesis $h$ such that $\Pr_{(x,y)\sim\D }\left[{h(x) \neq y}\right]\leq opt+\delta$.
\end{definition}

Note that when $\mathsf{opt} = 0$, this corresponds to the PAC model of learning.  Successful agnostic learning corresponds to learning in the presence of ``adversarial'' noise.

 The following lemma, considered folklore (see \cite{KlivansOS2004}), shows that
noise stable functions are well-approximated by low-degree polynomials.

\begin{lemma} \label{lem:lowdegree}
 Let $\Pi=\Pi_1\times \Pi_2 \times \cdots \times \Pi_n$ 
be a product distribution over $\oo^n$, and 
let $f:\oo^n\rightarrow \R$ be a function such that $\|f\|=1$ and $\NS_{\delta}(f)\leq \alpha(\delta)$ for some increasing function $\alpha:[0,1/2] \rightarrow [0,1]$.
Then there exists a multilinear polynomial $p:\oo^n\rightarrow\R$ of degree $\frac{1}{\alpha^{-1}(\delta/2.32)}$
such that
\[\E_{x\sim \Pi} \left[(f-p)^2\right] < \delta.\]
\end{lemma}

\ignore{The ``Low-Degree Algorithm'' \cite{LMN-1989,Mansour-1994} shows that by estimating
all the low-degree Fourier coefficients of a function with good Fourier concentration
one can learn the function to high accuracy.
The Low-Degree Algorithm is known to also work under any constant-bounded 
product distributions \cite{FJS-1991} as well.
\begin{corollary}\label{cor:low-degree}
Let $\CCC$ be the class of non-negative submodular functions with $\norm{f}_2=1$ 
and let the distribution over inputs be any product distribution with minimum
probability bounded by $p_{min}$. 
Then the Low-Degree Algorithm
outputs a hypothesis $h$ such that
$\E\left[(f-h)^2\right]=O(\gamma)$ given random examples in time
$\poly(n^{2/(1-c)},1/\gamma)$ where $c=(2\rho - 1 + 2p_{min}( 1 - \rho))$.
\end{corollary}
}

The ``$L_1$ Polynomial Regression Algorithm'' due to Kalai et al.~\cite{KalaiKMS2008} 
shows that one can \emph{agnostically} learn low-degree polynomials. 
\begin{theorem}[\cite{KalaiKMS2008}] \label{thm:KKMS}
Fix distribution $\D$ on $X \times \R$ with marginal $\D_X$ on $X$.  Suppose that for any $f \in \C$,  $\E_{x\sim D_X} \left[(f-p)^2\right] < \delta^2$ for some degree $d$ polynomial $p$.  Then, with high probability, the $L_1$ Polynomial Regression Algorithm
outputs a hypothesis $h$ such that $\Pr_{(x,y)\sim\D}\left[h(x) \neq y\right]\leq opt+\delta$
in time $\poly(n^d/\delta)$.
\end{theorem}


\section{Invariance Principle for Polytopes}
Our main invariance principle for polytopes $\ptp(W,t)$ is as follows: 
\begin{theorem}[invariance principle for polytopes]\label{th:ipmain}
For any proper and hypercontractive distribution $\mu$ over $\reals$ and any $\epsilon$-regular $k$-polytope $\ptp$,
\begin{equation}\label{}
\left| \, \pr_{X \lfta \mu^n}\left[X \in \ptp\right] - \pr_{Y \leftarrow \NN}\left[Y \in \ptp\right] \,\right| \leq C \,c_\mu^2\,(\log^{8/5} k) \,(\epsilon\,\log(1/\epsilon))^{1/5}.
\end{equation}
\end{theorem}
The proof of the theorem can be divided into three parts. 
\begin{enumerate}
\item We establish an invariance principle for smooth functions on polytopes (\tref{th:ipsmooth}) using an extension of Lindeberg's method; \sref{sec:ipsmoothp} is devoted to proving this part. 
\item We prove that for random variables $A,B$ over $\reals^k$,
  closeness with respect to smooth functions and anti-concentration
  bounds for one of the variables imply closeness with respect to
  rectangles (\lref{lm:smtopoly}). To do so, we use a 
result of Bentkus \cite{Bentkus1990} on smooth approximations for the $l_\infty$ norm. 
\item We use a result of Nazarov \cite{Nazarov2003} on Gaussian surface area of polytopes to bound the Gaussian measure of ``$l_\infty$-neighborhoods'' of polytopes in $\reals^n$ (\lref{lm:acrect}).
\end{enumerate} \ignore{
$\pr_{x \lfta \NN} \left[ W^T x \in \rect(\theta+\lambda \ok)\setminus \rect(\theta)\right].$
Note that it is easy to obtain a bound of $O(k \lambda)$ for the above probability. The much better bound of $O(\sqrt{\log k} \, \lambda)$ follows from a result of Nazarov on the Gaussian surface area of polytopes.}


We begin by stating an {\em invariance principle for smooth functions} $\psi:\reals^k \rgta \reals$.  The proof is involved, making use of the randomized-blockwise-hybrid argument alluded to in the introduction. For clarity we present the proof in the next section (\sref{sec:ipsmoothp}).
\begin{theorem}[invariance principle for smooth functions]\label{th:ipsmooth}
For any proper and hypercontractive distribution $\mu$ over $\reals$ and any $\epsilon$-regular $W$ and smooth function $\psi:\reals^k \rgta \reals$, 
$$ \left|\,\ex_{X \lfta \mu^n}\left[\psi(W^T X)\right] - \ex_{Y \lfta \NN}\left[\psi(W^T Y)\right] \,\right| \leq C \,c_\mu^2\, \psil \,(\log^3 k) \,(\epsilon\,\log(1/\epsilon)).$$
\end{theorem}
The following lemma shows that for two random variables $A,B$ over $\reals^k$, closeness with respect to smooth functions and {\em anti-concentration bounds} for the variable $B$ imply closeness with respect to rectangles. Note that to use the lemma we do not need anti-concentration bounds for the random variable $A$.
\begin{lemma}[smooth approximation of AND]\label{lm:smtopoly}
  Let $A,B$ be two random variables over $\reals^k$ satisfying the following conditions:
\begin{itemize}
 \item There exists $\Delta \geq 0$ such that for all smooth functions $\psi:\reals^k \rgta \reals$, $$\left|\ex\left[\psi(A)\right] - \ex\left[\psi(B)\right]\right| \leq \Delta \,\psil.$$ 
\item  There exists a function $g_k: [0,1] \rgta [0,1]$ such that the following holds: 
$$ \forall \lambda \in [0,1], \sup_{\theta \in \reals^k}\, \left\{\pr\left[\,B \in \rect(\theta+\lambda \ok) \setminus \rect(\theta)\,\right]\right\} \leq g_k(\lambda).$$
\end{itemize}
Then, 
$$\forall \theta \in \reals^k, \lambda \in (0,1), \left| \pr\left[A \in \rect(\theta)\right] - \pr\left[B \in \rect(\theta)\right]\,\right| \leq C \Delta \frac{\log^3 k}{\lambda^4}+ C g_k(\lambda).$$
\end{lemma}

Finally, we use the following anti-concentration bound that follows from Nazarov's estimate on the Gaussian surface area of polytopes \cite{Nazarov2003}:

\begin{lemma}[anti-concentration bound for $l_\infty$-neighborhood of rectangles]\label{lm:acrect}
 For $0 < \lambda < 1$, and $W \in \reals^{n\times k}$ such that each
 of the columns have norm 1,
$$ \pr_{x \lfta \NN} \left[ W^T x \in \rect(\theta) \setminus \rect(\theta-\lambda\ok)\right] = O(\lambda\,\sqrt{\log k}).$$
\end{lemma}

We first prove \tref{th:ipmain} using the above three results and then prove Lemmas~\ref{lm:smtopoly} and \ref{lm:acrect} in Sections~\ref{sec:smtopoly} and \ref{sec:acrect}. \tref{th:ipsmooth} is then proved in \sref{sec:ipsmoothp}.

\begin{proof}[of \tref{th:ipmain}]
Let $X \lfta \mu^n$, $Y \lfta \NN$ and let random variables $A = W^T X$, $B = W^T Y$. Then, by  \lref{lm:acrect} and \tref{th:ipsmooth},
 \[ \pr\left[B \in \reals(\theta + \lambda \ok) \setminus \reals(\theta)\right] \leq C \sqrt{\log k}\, \lambda,\]
\[ \left|\ex\left[\psi(A)\right] - \ex\left[\psi(B)\right]\right| \leq C \,c_\mu^2\,(\log^3 k)\, \epsilon\,\log(1/\epsilon)\, \psil,\] where $\psi:\reals^k \rgta \reals$ is any smooth function, $\theta \in \reals^k$ and $\lambda \in (0,1)$. Therefore, by \lref{lm:smtopoly}, for $\theta \in \reals^k$, 
\[\left|\pr\left[A \in \rect(\theta)\right] - \pr\left[B \in \rect(\theta)\right] \right| \leq C \,(\log^6 k)\,\log(1/\epsilon)\epsilon/\lambda^4 + C \sqrt{\log k}\, \lambda. \]
The theorem now follows by setting $\lambda = (\log^{11/10} k)\,(\epsilon\log(1/\epsilon))^{1/5}$.
\end{proof}

\subsection{Smooth approximation of AND}\label{sec:smtopoly}
We now prove \lref{lm:smtopoly}. For this, we use the following 
result of Bentkus \cite{Bentkus1990} on smooth approximations for the $l_\infty$ norm.
\begin{theorem}[Bentkus \cite{Bentkus1990}]\label{bentkus}
  For every $\alpha > 0$ and $ 0 < \lambda < 1$, there exists a function $\psi \equiv \psi_{\alpha,\lambda}:\reals^k \rgta \reals$ such that $\psil \leq C \log^3 k /\lambda^4$ and
\[ \psi(a) =
\begin{cases}
  1 & if\;\; \|a\|_\infty \leq \alpha\\
  0 & if \;\; \|a \|_\infty > \alpha+\lambda\\
  \in [0,1] & otherwise
\end{cases}.\]
\end{theorem}

\begin{corollary}\label{bentkus1}
For all $u \in \reals^k$, $0 < \lambda < 1$, $T > \|u\|_\infty$, there exists a function $\psi \equiv \psi_{u,\lambda,T}:\reals^k \rgta \reals$ such that $\psil \leq C \log^3 k /\lambda^4$ and 
\[ \psi(a) =
\begin{cases}
  1 & if\;\; \forall l \in [k] , -T + u_l \leq a_l \leq u_l\\
  0 & if \;\;  \exists l \in [k] , a_l > u_l+\lambda\\
  \in [0,1] & otherwise
\end{cases}.\]
\end{corollary}
\begin{proof}
Let $\psi_{T/2,\lambda}$ be the function from \tref{bentkus} with $\alpha = T/2$. Define $\psi \equiv \psi_{u,\lambda,T}:\reals^k \rgta \reals$ by 
\[ \psi_{u,\lambda,T}(a_1,\ldots,a_k) = \psi_{T/2,\lambda}(a_1 + T/2 - u_1, a_2 + T/2 - u_2,\ldots,a_k + T/2 - u_k).\]
It is easy to check that $\psi$ satisfies the conditions of the theorem.
\end{proof}

\begin{proof}[of \lref{lm:smtopoly}]
Fix $\theta \in \reals^k$, $0 < \lambda < 1$. Choose $T \in \reals$
large enough so that $T > \|\theta\|_\infty$, $\pr\left[\|A\|_\infty
  \geq T\right] < \Delta $ and $\pr\left[\|B\|_\infty \geq T\right] <
\Delta$. Then, by the choice of $T$
  \begin{align}\label{ip:eq3}
\left| \pr\left[A \in \rect(\theta)\right] - \pr\left[A \in \rect_{2T}(\theta)\right] \right| \leq \Delta,\nonumber\\
\left|\pr\left[B \in \rect(\theta)\right] - \pr\left[B \in \rect_{2T}(\theta)\right]\right| \leq \Delta,     
  \end{align}
where $\rect_T(\theta) = [-T+\theta_1,\theta_1] \times [-T+\theta_2,\theta_2] \times \cdots \times [-T+\theta_k,\theta_k] \subseteq \reals^k$.  
Let $\psi:\reals^k \rgta \reals$ be the function obtained from
applying \cref{bentkus1} to $\theta, \lambda, 2T$. Observe that from
the definition of $\psi$ in \cref{bentkus1} and \eref{ip:eq3}, we have
\[ \pr\left[A \in \rect(\theta)\right] \leq \ex\left[\psi(A)\right] + \Delta \leq \ex\left[\psi(B)\right] + \Delta \psil + \Delta .\]
Similarly, 
\begin{align*}
\ex\left[\psi(B)\right] &\leq \pr\left[B \in \rect(\theta+\lambda\ok)\right] \\
&= \pr\left[B \in \rect(\theta)\right] + 
\pr\left[B \in \rect(\theta+\lambda\ok)\setminus\rect(\theta)\right] \\&\leq \pr\left[B \in \rect(\theta)\right] + g_k(\lambda),
\end{align*}
where the last inequality follows from the definition of $g_k$. Combining the above two equations we get
\begin{multline*}
 \pr\left[A \in \rect(\theta)\right] \leq \pr\left[B \in \rect(\theta)\right] + 2 \Delta \psil + g_k(\lambda) \leq \\
\pr\left[B \in \rect(\theta)\right] + \frac{C \Delta \log^3 k}{\lambda^4} + g_k(\lambda).
\end{multline*}

Proceeding similarly for the function $\psi_L:\reals^k \rgta \reals$ obtained by applying \cref{bentkus1} to $\theta-\lambda\ok,\lambda,2T$, we get 
\begin{equation*}
 \pr\left[A \in \rect(\theta)\right] \geq \pr\left[B \in \rect(\theta)\right] - \frac{C \Delta \log^3 k}{\lambda^4} - g_k(\lambda).
 \end{equation*}
Therefore, 
\[ \left|\pr\left[A \in \rect(\theta)\right] - \pr\left[B \in \rect(\theta)\right]\right| \leq \frac{C \Delta \log^3 k}{\lambda^4} + g_k(\lambda).\]
\end{proof}

\subsection{Anti-concentration bound for $\l_\infty$-neighborhood of rectangles}\label{sec:acrect}

\lref{lm:acrect} follows straightforwardly from the following result of Nazarov \cite{Nazarov2003}. For a convex body $K \subseteq \reals^n$ with boundary $\partial K$, let $\Gamma(K)$ denote the Gaussian surface area of $K$ defined by 
\[ \Gamma(K) = \int_{y \in \partial K} e^{\frac{-\|y\|^2}{2}}\, d\sigma(y),\]
where $d\sigma(y)$ denotes the surface element at $y \in \partial K$.
\begin{theorem}[Nazarov (see {\cite[Theorem~20]{KlivansOS2008}})]\label{th:nazarov}
For a polytope $\ptp$ with at most $k$ faces,  $\Gamma(\ptp) \leq C \sqrt{\log k}$.
\end{theorem}
\begin{proof}[of \lref{lm:acrect}]
Consider an increasing (under set inclusion) family of polytopes $\ptp_\rho$ for $0 \leq \rho \leq \lambda$ such that $\ptp_0 = \{x: W^Tx \in \rect(\theta-\lambda\ok)\}$ and $\ptp_\lambda = \{x:W^Tx \in \rect(\theta)\}$. Then, 
\[ \pr_{x \lfta \NN} \left[ W^T x \in \rect(\theta)\setminus \rect(\theta-\lambda\ok)\right] = \int_{\rho = 0}^\lambda \Gamma(\ptp_\rho) d\rho \leq C \sqrt{\log k}\, \lambda,\]
where the last inequality follows from \tref{th:nazarov}.
\end{proof}



\section{Invariance principle for Smooth Functions over Polytopes}\label{sec:ipsmoothp}
We now prove \tref{th:ipsmooth}. The proof of the theorem is based on
the Lindeberg method for proving limit theorems with explicit error
bounds. Let $t = \lceil 1/\epsilon\rceil $ 
and let $\hh = \{h:[n] \rgta [t]\}$ be a family of $(2\log k)$-wise independent functions. That is, for all $I \subseteq [n], |I| \leq 2\log k$ and $b \in [t]^{I}$,  $\pr_{h \in_u \hh}[\, \forall i \in I,\; h(i) = b_i\,] = \frac{1}{t^{|I|}}$.

We remark that to prove \tref{th:ipsmooth} we could take the hash family to be the set of all functions. However, we work with a $(2\log k)$-wise independent family as the analysis is no more complicated and we need to work with such hash families while constructing pseudorandom generators. For $S \subseteq [n]$, let $W_S$ be the matrix formed by the rows of $W$ with indices in $S$ (thus $W^{i}_{S}$ is the $i$th column of the submatrix whose rows are given by the indices in $S$). Define
\[ \hh(W) \eqdef \sum_{i=1}^t \left(\,  \ex_{h}\left[\,\sum_{p=1}^k \|W_{h^{-1}(i)}^p\|^{4\log k}\right]\,\right)^{1/\log k}.\]  

\tref{th:ipsmooth} follows immediately from the following two lemmas. 
\begin{lemma}\label{lm:hhw}
  For $\epsilon$-regular $W$, $\hh(W) \leq C \,\log k\, (\epsilon\, \log(1/\epsilon))$. 
\end{lemma}

\begin{lemma}\label{lm:foolsmooth}
For any smooth function $\psi:\reals^k \rgta \reals$, 
$$\left| \ex_{X \lfta \mu^n}\left[ \psi(W^T X)\right] - \ex_{Y \lfta \NN}\left[\psi(W^T Y)\right] \right| \leq 4\,c_\mu^2\, (\log^2 k) \,\hh(W) \psil.$$
\end{lemma}

\begin{proof}[of \lref{lm:hhw}]
Fix a $l \in [t]$, $p \in [k]$. For $i \in [n]$, let $X_i$ be the indicator random variable that is $1$ if $h(i) = l$ and $0$ otherwise. Then, $\pr[X_i = 1] = 1/t$ and the variables $X_1,\ldots,X_n$ are $(2\log k)$-wise independent. Further, 
\[  Z_p' \equiv \|W^p_{|h^{-1}(l)}\|^{2} = \sum_{i=1}^n W_{ip}^2 X_i .\]
Let $Y_i$ be i.i.d indicator random variables with $\pr[Y_i = 1] = 1/t$ and let $Z_p = \sum_{i=1}^n W_{ip}^2 Y_i$. Observe that $Z_p'$ and $Z_p$ have identical $d$'th moments for $d \leq 2\log k$. Moreover, by Hoeffding's inequality applied to $Z_p$, for any $\gamma > 0$,
\begin{equation*}
  \label{eq:sm2}
  \pr\left[\, \left|Z_p - \frac{1}{t}\right| \geq \gamma \,\right] \leq 2 \exp\left(-\frac{2\gamma^2}{\sum_{i=1}^n W_{ip}^4}\right) \leq 2 \exp\left(-\frac{2\gamma^2}{\epsilon^2}\right) = 2 \exp(-2t^2\gamma^2).
\end{equation*}
The above tail bound for $Z_p$ implies strong bounds on the moments of $Z_p$ by standard arguments. Setting $\gamma = \sqrt{2 \log k \,\log t}/t$ in the above equation, we get
\begin{equation*}
  \label{eq:sm6}
  \pr\left[\, |Z_p| \geq \frac{\sqrt{3 \log k\,\log t}}{t}\right ] \leq \frac{1}{t^{2\log k}}.
\end{equation*}
Therefore, from the above equation and the fact that $Z_p \leq 1$
\begin{align*}
  \ex\left[Z_p^{2\log k}\right] &\leq \frac{(3 \log k\,\log t)^{\log k}}{t^{2\log k}} + \pr\left[\, |Z_p| \geq \frac{\sqrt{3 \log k\,\log t}}{t}\right ]\\
&\leq  \frac{(4 \log k\,\log t)^{\log k}}{t^{2\log k}}.
\end{align*}

Therefore, 
\begin{equation*}
  \label{eq:sm4}
  \ex_{h \in_u \hh}\left[\,\|W^p_{|h^{-1}(l)}\|^{4\log k}\,\right] = \ex\left[\,(Z_p')^{2\log k}\right] = \ex\left[\,Z_p^{2\log k}\,\right] \leq \frac{(4\log k \log t)^{\log k}}{t^{2\log k}}.
\end{equation*}
Therefore, from the definition of $\hh(W)$ and the above equation, 
\begin{multline*}
 \hh(W) = \sum_{i=1}^t \left(\,  \sum_{p=1}^k \ex_{h}\left[\,\|W_{h^{-1}(i)}^p\|^{4\log k}\right]\,\right)^{1/\log k} \leq t \frac{4 \log k \log t}{t^2} =\\ 4 (\log k)(\epsilon \log(1/\epsilon)).  
\end{multline*}
\end{proof}
The proof of \lref{lm:foolsmooth} uses a blockwise hybrid argument and careful applications of hypercontractivity as sketched in the proof outline in the introduction. To gain some intuition of the advantage of our randomized blockwise hybrid argument over the standard Lindeberg method, it might be helpful to compare both arguments for the following cases:\\

{\bf Example 1}: The bounding hyperplanes of $\ptp$ are oriented majorities: $W \in \{1/\sqrt{n}, -1/\sqrt{n}\}^{n \times k}$. In this case, the standard Lindeberg method in conjunction with Bentkus's smoothing function and Nazarov's surface area bound as used in Lemmas \ref{lm:smtopoly}, \ref{lm:acrect} can be adapted (without having to do a blockwise hybrid argument) to get a bound as in \tref{th:ipmain}.\\

{\bf Example 2}: The bounding hyperplanes of $\ptp$ are oriented majorities on disjoint sets of variables: For $m = n/k$ and each $p \in [k]$, $m = n/k$, $W^p_i = 1/\sqrt{m}, (p-1)m + 1 \leq i \leq pm$ and $W^p_i = 0$ otherwise. In this case, however, when $m \geq 1/\epsilon^2$ (so each bounding hyperplane is still regular), it is easy to see that the standard Lindeberg method (even when used in conjunction with Lemmas \ref{lm:smtopoly}, \ref{lm:acrect}) leads to an error bound that is at least linear in $k$.\\

We use the following form of the standard Taylor series expansion (the interested reader can find more about the multivariate Taylor theorem on the wikipedia page for ``Taylor's Theorem''). For a smooth function $\psi:\reals^k \rgta \reals$, $x \in \reals^k$ and $p_1,\ldots,p_r \in [k]$, let $\partial_{p_1,\ldots,p_r} \psi(x) = \partial_{p_1}\partial_{p_2}\cdots\partial_{p_r} \,\psi(x)$. For indices $p_1,\ldots,p_r \in [k]$, let $(p_1,\ldots,p_r)! = s_1! s_2 ! \ldots s_k!$, where, for $l \in [k]$,  $s_l$ denotes the number of occurrences of $l$ in $(p_1,\ldots,p_r)$. 
\begin{fact}[Multivariate Taylor's Theorem]\label{lm:taylor}
  For any smooth function $\psi:\reals^k \rgta \reals$, and $x, y \in \reals^k$, 
  \begin{multline*}
 \psi(x+y) = \psi(x) + \sum_{p \in [k]} \partial_p \psi(x)\, y_p + \sum_{p,q \in [k]} \frac{1}{(p,q)!}\, \partial_{p,q} \psi(x)\, y_p y_q 
  + \\\sum_{p,q,r \in [k]} \,\frac{1}{(p,q,r)!}\, \partial_{p,q,r} \psi(x)\, y_p y_q y_r + \mathsf{err}(x,y),
  \end{multline*}
where $|\mathsf{err}(x,y)| \leq \psil \cdot \max_{p \in [k]} |y_p|^4$.
\end{fact}

\begin{proof}[of \lref{lm:foolsmooth}]
Let $\overline{X} \lfta \mu^n$ and $\overline{Y} \lfta \NN$. We first partition $[n]$ into blocks using a random hash function $h \in_u \hh$ and then use a blockwise-hybrid argument. Fix a hash function $h \in \hh$. View $\overline{X}$ as $X^1,\ldots,X^t$, where each $X^l = \overline{X}_{h^{-1}(l)}$ is chosen independently and uniformly from $\mu^{|h^{-1}(l)|}$. Similarly, view $\overline{Y}$ as $Y^1,\ldots,Y^t$ where each $Y^l = \overline{Y}_{h^{-1}(l)}$ is chosen independently and uniformly from $\mathcal{N}^{|h^{-1}(l)|}$. We prove the claim via a hybrid argument where we replace the blocks $X^1,\ldots,X^t$ with $Y^1,\ldots,Y^t$ one at a time.

For $0 \leq i \leq t$, let $Z^i$ be the distribution with $Z^i_{|h^{-1}(j)} = X^j$ for $i < j \leq t$ and $Z^i_{|h^{-1}(j)} = Y^j$ for $1 \leq j \leq i$. Then, $Z^0$ is distributed as $\mu^n$ and $Z^{t}$ is distributed as $\NN$. For $l \in [t]$, let  
\[ h(W,l) = \left( \sum_{p = 1}^k \|W^p_{h^{-1}(l)}\|^{4\log k}\right)^{1/\log k}.\]
\begin{claim}
  For $1 \leq l \leq t$, and fixed $h \in \hh$, $$ \left|\,\ex_{\overline{X},\overline{Y}}\left[\psi(W^TZ^l)\right] - \ex_{\overline{X},\overline{Y}}\left[\psi(W^T Z^{l-1})\right]\,\right| \leq C\,c_\mu\,\log^2 k\, \psil\,h(W,l).$$
\end{claim} 
\begin{proof}
    Without loss of generality, suppose that $h^{-1}(l) = \{1,\ldots,m\}$. Note that $Z^l,Z^{l-1}$ have the same random variables in positions $m+1,\ldots,n$. Let $Z^{l-1} = (X_1,\ldots,X_m,Z_{m+1},\ldots,Z_n)$ and 

\noindent $Z^{l} = (Y_1,\ldots,Y_m,Z_{m+1},\ldots,Z_n)$ where $(X_1,\ldots,X_m)$ is uniform over $\mu^m$ and $(Y_1,\ldots,Y_m) $ is uniform over $\ugs^m$. Note that $(Z_{m+1},\ldots,Z_n)$ is independent of $(X_1,\ldots,X_m)$, $(Y_1,\ldots,Y_m)$. 

 Let $W_1 \in \reals^{m\times k}$ be the matrix formed by the first $m$ rows of $W$ and similarly let $W_2 \in \reals^{(n-m) \times k}$ be the matrix formed by the last $n-m$ rows of $W$. Lastly, let $V= W_2^T (Z_{m+1},\ldots,Z_n)$ and $U $ be one of $X = (X_1,\ldots,X_m)$ or $Y = (Y_1,\ldots,Y_m)$. Now, by using a Taylor expansion of $\psi$ at $V$ as in \newref[Fact]{lm:taylor},

\begin{multline}\label{prg:eq1}
  \psi(W^T (U_1,\ldots,U_m,Z_{m+1},\ldots,Z_n)) = \psi(\,W_1^T U + V\,)\\
= \psi(V) + \sum_{p \in [k]} \partial_p \psi(V) \,\iprod{W^p_1}{U} + \,\sum_{p,q \in [k]}\frac{1}{(p,q)!}\, \partial_{p,q} \psi(V)\, \iprod{W^p_1}{U}\,\iprod{W^q_1}{U}\\
+ \sum_{p,q,r \in [k]}\frac{1}{(p,q,r)!}\, \partial_{p,q,r} \psi(V) \,\iprod{W^p_1}{U} \,\iprod{W^q_1}{U} \,\iprod{W^r_1}{U}  + \mathsf{err}(V,W_1^TU).
\end{multline}

Now, using the fact that $\|z\|_{\infty} \leq \|z\|_{\log k}$ for $z \in \reals^{k}$,
\begin{equation}
  \label{prg:eq2} 
  \left|\mathsf{err}(V,W_1^TU)\right| \leq \psil \cdot \max_{p \in [k]} |\iprod{W^p_1}{U}|^4 \leq \psil \left(\sum_{p=1}^k |\iprod{W^p_1}{U}|^{4\log k}\right)^{1/\log k}.
\end{equation}

Now, by hypercontractivity of $\mu$, 
\begin{align}\label{ips:eq4}
\ex_X\left[  \left(\sum_{p=1}^k |\iprod{W^p_1}{X}|^{4\log k}\right)^{1/\log k}\right] &\leq \left(\ex_X\left[\sum_{p=1}^k |\iprod{W^p_1}{X}|^{4\log k}\right]\right)^{1/\log k}\nonumber\\
& \text{\hspace{0.5in} (by power-mean inequality)}\nonumber\\
&= \left(\sum_{p=1}^k \ex_X\left[|\iprod{W^p_1}{X}|^{4\log k}\right] \right)^{1/\log k}\nonumber\\
&\leq  \left(\sum_{p=1}^k (c_\mu \log k)^{2\log k} \,\|W^p_1\|^{4\log k}\right)^{1/\log k} \nonumber\\
&\text{\hspace{0.2in} (by hypercontractivity of $\mu$)}\nonumber\\
&\leq C c_\mu^2(\log^2 k) \, h(W,l).
\end{align}

Similarly, by hypercontractivity of $\mathcal{N}$, 
\begin{equation}
  \label{ips:eq5}
  \ex_Y\left[  \left(\sum_{p=1}^k |\iprod{W^p_1}{Y}|^{4\log k}\right)^{1/\log k}\right] \leq C (\log^2 k) \, h(W,l).
\end{equation}

Since $\mu$ is proper, for any $u^1,u^2,u^3 \in \reals^m$,
\[ \ex\left[\iprod{u^1}{X}\right] = \ex\left[\iprod{u^1}{Y}\right],\;\;\; \ex\left[\iprod{u^1}{X}\,\iprod{u^2}{X}\right] = \ex\left[\iprod{u^1}{Y}\,\iprod{u^2}{Y}\right]\]
\[\ex\left[\iprod{u^1}{X}\,\iprod{u^2}{X}\,\iprod{u^3}{X}\right] = \ex\left[\iprod{u^1}{Y}\,\iprod{u^2}{Y}\,\iprod{u^3}{Y}\right].\]

\ignore{
\begin{align*}
\ex[\iprod{u^1}{X}] &= \ex\left[\iprod{u^1}{Y}\right] \\
\ex[\iprod{u^1}{X}\,\iprod{u^2}{X}] &= \ex\left[\iprod{u^1}{Y}\,\iprod{u^2}{Y}\right]\\
\ex[\iprod{u^1}{X}\,\iprod{u^2}{X}\,\iprod{u^3}{X}] &= \ex\left[\iprod{u^1}{Y}\,\iprod{u^2}{Y}\,\iprod{u^3}{Y}\right].
\end{align*}}
From the above equations, Equations \eqref{prg:eq1}, \eqref{prg:eq2}, \eqref{ips:eq4}, \eqref{ips:eq5} and the fact that $X,Y,V$ are independent of one another, it follows that 
\[\left|\ex\left[\psi(W^T Z^l) - \psi(W^TZ^{l-1})\right] \right| \leq C c_\mu^2(\log^2 k)\, \psil h(W,l).\]
\end{proof}
\lref{lm:foolsmooth} now follows from the above claim, summing from $l=1,\ldots,t$, and taking expectation with respect to $h \in_u \hh$.
\end{proof}

\section{Lowerbound for Error}\label{sec:lbound}
We now show that \tref{th:ipmain} is essentially optimal by showing that the error bound in any such result cannot be $o(\epsilon \cdot \sqrt{\log k})$. We do so by constructing a $(1/\sqrt{\log k})$-regular $k$-polytope for which the error is $1 - o(1)$.

Let $k = 2^r$ and let $\ptp = \{x: \sum_{i=1}^r |x_i| < r\} \subseteq \reals^r$ be the $\ell_1$-ball of dimension $r$ and radius $r$. We will use the following simple facts.

\begin{fact}
  $\ptp$ is a $(1/\sqrt{r})$-regular $k$-polytope.
\end{fact}
\begin{proof}
  Note that $\ptp = \{x \in \reals^r: \iprod{x}{y} < r, \forall y \in \dpm^r\}$. As $y \in \dpm^r$ are $(1/\sqrt{r})$-regular, the claim follows.
\end{proof}
\begin{fact}
  For $X \in_u \dpm^r$, $Y \lfta \mathcal{N}^r$, $|\pr\left[X \in \ptp\right] - \pr\left[Y \in \ptp\right]| = 1 - o(1)$.
\end{fact}
\begin{proof}
  Clearly, $\pr\left[X \in \ptp\right] = 0$. We next show that $\pr\left[Y \in \ptp\right] = 1 - \exp(\Omega(r))$. Observe that $Y \in \ptp$ if and only if $\|Y\|_1 < r$. By linearity of expectation, 
$$\ex\left[\|Y\|_1\right] = r \cdot \ex_{y \lfta \mathcal{N}(0,1)}\left[ |y|\right] = r c,$$ 
where $c=\sqrt{2/\pi}$ is a constant strictly less than $1$. Note that by Cauchy-Schwarz, the $\ell_1$ norm has Lipschitz constant $\sqrt{r}$. Therefore, by \tref{th:gaussianld},
\[ \pr\left[ \|Y\|_1 \geq r\right] \leq \pr\left[\, |\|Y\|_1 - c r| \geq (1-c)r\,\right] \leq 2 \exp(-(1-c)^2 r^2/r) = \exp(-\Omega(r)).\]
Thus, $\pr\left[Y \in \ptp\right] = \pr\left[\|Y\|_1 < r\right] = 1 - \exp(-\Omega(r))$. The claim now follows.
\end{proof}

The above two claims show that any invariance principle as in \tref{th:ipmain} must incur an error of $\Omega(\epsilon \sqrt{\log k})$, which matches the bound of \tref{th:ipmain} up to a polylogarithmic factor in $k$ and a polynomial factor in $\epsilon$.
\section{Noise Sensitivity of Intersections of Regular Halfspaces}
We now describe how our invariance principle yields a bound on the
average and noise sensitivity of intersections of regular
halfspaces. (see \dref{def:noisesen} for definition of (Boolean) noise
sensitivity). 

Let $f^1,\ldots,f^k:\dpm^n \rgta \dpm$ be halfspaces with $f^p(x) = \sign(\iprod{W^p}{x} - \theta_p)$ and let $\fak:\dpm^n \rgta \dpm$ be their intersection, $\fak = f^1 \wedge f^2 \wedge \ldots \wedge f^k$. 

\begin{theorem} \label{thm:ns}
For $\fak$ $\epsilon$-regular, $\NS_\delta(\fak) \leq C (\log^{1.6} (k/\delta))\,(\epsilon^{1/6} + \delta^{1/2})$.
\end{theorem}
We prove the theorem by first reducing bounding noise sensitivity of $\fak$ to bounding the Boolean volume of $l_\infty$-neighborhoods of polytopes. We then use our invariance principle, \tref{th:ipmain}, to prove the required bounds on the Boolean volume of boundaries of polytopes. 

As mentioned before, the above theorem implies a $n^{\log^{O(1)}k}$ algorithm for learning intersections of regular halfspaces in the agnostic model for any constant error rate.


We use the following tail bound that follows from Pinelis's subgaussian tail estimates~\cite{Pinelis1994}.

\begin{fact}\label{lm:stail}
There exist absolute constants $c_1,c_2 > 0$ such that all $w \in \reals^m$, $t > 0$, $$\pr_{x \in_u \dpm^m} \left[\,|\iprod{w}{x}| > t \|w\|\,\right] \leq c_1 \exp(-c_2 t^2).$$
\end{fact}

The following claim says that for $W$ $\epsilon$-regular, random $x \in_u \dpm^n$, and a $\delta$-perturbation $y$ of $x$, $W^Tx$ is close to $W^Ty$ in $l_\infty$ distance.

\begin{claim}\label{clm:ns5}
For $x \in \dpm^n$, let $y(x)$ be a random $\delta$-perturbation of $y(x)$ of $x$. Then, 
\[ \pr_{x \in_u \dpm^n,y(x)}\left[ \,\|W^Tx - W^T y(x)\|_\infty \geq \lambda \,\right] \leq 2 \delta,\]
where $\lambda = C \log(k/\delta)^{1/2} \delta^{1/2} + C \log (k/\delta)^{3/4} \epsilon^{1/2}$.
\end{claim}
\begin{proof}
  Let $Y = (Y_1,\ldots,Y_n)$ be i.i.d indicator variables with $\pr\left[Y_i = 1\right] = \delta$. Let $S(Y) = support(Y)$. Now, for $p \in [k]$, $\|W^p_{S(Y)}\|^2 = \sum_{i = 1}^n W_{ip}^2 Y_i$ and $\ex\left[\|W^p
_{S(Y)}\|^2\right] = \delta$. Further, since $W$ is $\epsilon$-regular, by Hoeffding's inequality, for all $t > 0$,
\[ \pr\left[\,|\|W^p_{S(Y)}\|^2 -\delta|\,\geq\, \gamma \right] \leq 2\, \exp\left(\frac{-2\gamma^2}{\sum_i W_{ip}^4}\right) \leq 2 \,\exp\left(\frac{-2\gamma^2}{\epsilon^2}\right).\]
Thus, by a union bound
\begin{equation}
  \label{eq:ns1}
   \pr_Y\left[\,\exists p \in [k],\,\|W^p_{S(Y)}\|^2 \geq \delta + 2 \sqrt{\log(k/\delta)}\, \epsilon\,\right] \leq \delta.
\end{equation}

Note that for a fixed $Y$ and sufficiently large $C$, by \newref[Fact]{lm:stail} and a union bound,
\[ \pr_{x \in_u \dpm^n}\left[\,\exists p \in [k],\,|\iprod{W^p_{S(Y)}}{x_{S(Y)}}| \geq C\sqrt{\log (k/\delta)}\, \|W^p_{S(Y)}\|\,\right] \leq \delta.\]
From \eref{eq:ns1} and the above equation, we get that for a sufficiently large constant $C$
\begin{equation}
  \label{eq:ns4}
 \pr_{x \in_u \dpm^n,Y}\left[\,\exists p \in [k],\,|\iprod{W^p_{S(Y)}}{x_{S(Y)}}| \geq C \log(k/\delta)^{1/2} \delta^{1/2} + C \log (k/\delta)^{3/4} \epsilon^{1/2} \,\right] \leq 2 \delta.  
\end{equation}
Now, observe that that for $x \in \dpm^n$, to generate a $\delta$-perturbation of $x$, $y(x)$, we can first generate a random $Y$ as above and flip the bits of $x$ in the support of $Y$. Thus, from \eref{eq:ns4},
\begin{multline*}\label{}
 \pr_{x \in_u \dpm^n,Y}\left[\,\exists p \in [k]\,|\iprod{W^p}{x} - \iprod{W^p}{y(x)}| \geq \lambda \,\right] =\\  \pr_{x \in_u \dpm^n,Y}\left[\,\exists p \in [k]\,|\,\iprod{W^p_{S(Y)}}{x_{S(Y)}}| \geq \lambda\,\right] \leq 2 \delta,
\end{multline*}
where $\lambda = C \log(k/\delta)^{1/2} \delta^{1/2} + C \log (k/\delta)^{3/4} \epsilon^{1/2}$.
Therefore, 
\begin{equation*}
  \label{eq:ns5}
  \pr_{x \in_u \dpm^n,Y}\left[ \,\|W^Tx - W^T y(x)\|_\infty \geq \lambda \,\right] \leq 2 \delta.
\end{equation*}
\end{proof}

The following claim can be seen as an anti-concentration bound for regular polytopes over the hypercube and may be of independent interest:

\begin{claim}\label{clm:ns6}
For $\epsilon$-regular $W \in \rnk$, $\theta \in \reals^k$, and $0 < \lambda < 1$,
\begin{multline*} \pr_{x \in_u \dpm^n}\left[\,W^T x \in \rect(\theta + \lambda \, \ok) \setminus \rect(\theta - \lambda \,\ok)\,\right] \leq \\ C (\log^{1.6} k)\,(\epsilon\,\log(1/\epsilon))^{1/5} +  \sqrt{\log k}\, \lambda.
\end{multline*}
\end{claim}
\begin{proof}
Follows directly from \tref{th:ipmain} and \lref{lm:acrect}.
\end{proof}

We can now prove \tref{thm:ns}.
\begin{proof}[of \tref{thm:ns}]
Note that for $x,  y \in \reals^n$, $\fak(x) \neq \fak(y)$ implies that $W^T x \in \rect(\theta + \gamma \ok) \setminus \rect(\theta - \gamma\ok)$, where $\gamma = \|W^Tx - W^T y\|_\infty$. Hence,
\begin{align*}
  \NS_\delta(\fak) &= \pr_{x \in_u \dpm^n,Y}\left[\,\fak(x) \neq \fak(y(x))\,\right]\\
  &\leq \pr_{x \in_u \dpm^n,Y}\left[\,\fak(x) \neq \fak(y(x))\,|\,\|W^Tx - W^T y(x)\|_\infty \leq \lambda\,\right] + 2 \delta \\
&\text{\hspace{.5in} (\clref{clm:ns5})}\\
  &\leq \pr_{x \in_u \dpm^n}\left[\,W^T x \in \rect(\theta + \lambda \, \ok) \setminus \rect(\theta - \lambda \,\ok)\,\right] + 2\delta\\
  &\leq C (\log^{1.6} k)\,(\epsilon\,\log(1/\epsilon))^{1/5} +  \sqrt{\log k}\, \lambda + 2 \delta.\\
& \text{\hspace{.5in} (\clref{clm:ns6})}
\end{align*}
The theorem now follows.
\end{proof}

Applying \lref{lem:lowdegree} and \tref{thm:KKMS} with \tref{thm:ns}, we immediately obtain our main result for learning intersections of halfspaces, namely \tref{thm:mainlearnstate}.




\section{Pseudorandom Generators for Polytopes} 
We now prove our main theorems for constructing pseudorandom generators for polytopes with respect to a variety of distributions (Theorems \ref{th:prgmain}, \ref{th:prgnormal},  and \ref{th:prgspherical}).




The results in this section are based on a recent PRG construction due to Meka and Zuckerman\cite{MekaZ2010} for polynomial threshold functions using the invariance principle of Mossel et al.~\cite{MosselOO2005}.  A closer look at their construction reveals a general program for constructing PRGs from invariance principles.  Given this observation, it is natural to ask if our invariance principle can be used to construct PRGs for regular polytopes.  Indeed it can, and we use the Meka and Zuckerman generator but with a different setting of its parameters.  The analysis, however, is a little more complicated in our setting (even given our invariance principle) and requires a careful application of hypercontractivity.



\subsection{Main Generator Construction}
We begin by describing the construction of the PRG we use; it is a
slightly modified version of the PRG used by \cite{MekaZ2010} to fool regular {\em halfspaces} (i.e., the case $k=1$).

Give $\delta \in (0,1)$, let $\epsilon = \Omega(\delta^6/\log^{9.6} k)$ be such that $\log^{1.6} k (\epsilon \log(1/\epsilon))^{1/5} = \delta$.
Let $t = 1/\epsilon$ and let $\hh = \{h: h:[n] \rgta [t]\}$ be a $(2\log k)$-wise independent family of hash functions. That is, for all $I \subseteq [n], |I| \leq 2\log k$ and $b \in [t]^{I}$, \[\pr_{h \in_u \hh}\left[\, \forall i \in I,\; h(i) = b_i\,\right] = \frac{1}{t^{|I|}}.\]
Efficient constructions of hash families $\hh$ as above with $|\hh| = O(n^{2\log k})$ are known. To avoid some technical issues that can be overcome easily, we assume that every hash function $h \in \hh$ is equi-distributed in the following sense: for all $j \in [t]$, $|\{i:h(i) = j\}| = n/t$. 

Let $m = n/t$ and let $G_0:\zo^s \rgta \dpm^m$ generate a $(4\log k)$-wise independent distribution over $\dpm^m$. That is, for all $I \subseteq [n], |I| \leq 2\log k$ and $b \in \dpm^{I}$, \[\pr_{x = G_0(z), z \in_u \zo^s}\left[\,\forall i \in I,\; x_i = b_i\,\right] = \frac{1}{2^{|I|}}.\]  Efficient constructions of generators $G_0$ as above with $s = O(\log k\,\log n)$ are known \cite{NaorN1993}.

Given a hash family and generator $G_0$ as above, we consider the following generator. Define $G:\hh \times (\zo^s)^t \rgta \dpm^n$ by $$\greg(h,z^1,\ldots,z^t) = x, \text{ where $x_{|h^{-1}(i)} = G_0(z^i)$ for $i \in [t]$.}$$

\subsection{Pseudorandom Generators for Regular Polytopes} \label{sec:prgreg}
We now argue that the generator $G$ defined in the last section fools regular polytopes and prove \tref{th:prgmain}. 
\begin{proof}[of \tref{th:prgmain}]
The bound on the seed length of the generator $G$ follows from the construction. The following statement follows from an argument similar to that of the proof of \tref{th:ipsmooth}: for any smooth function $\psi:\reals^k \rgta \reals$ and $\epsilon$-regular $W$,
\begin{equation}
  \label{eq:prgrp1}
 \left|\,\ex_{y \in_u \zo^r}\left[\psi(W^T G(y))\right] - \ex_{Y \lfta \NN}\left[\psi(W^T Y)\right] \,\right| \leq C \log^3 k \,(\epsilon\log(1/\epsilon))\, \psil.  
\end{equation}
Indeed, to observe that \lref{lm:hhw} holds for any $(2\log k)$-wise independent family of hash functions and the proof of \lref{lm:foolsmooth} relies only on two key properties of $X \lfta \mu^n$: (1) For a fixed hash function $h$, the blocks $X_{h^{-1}(1)},X_{h^{-1}(2)},\ldots,X_{h^{-1}(t)}$ are independent of one another. (2) For a fixed hash function $h$, and $j \in [t]$, the distribution of each block $X_{h^{-1}(j)}$ satisfies $(2,2\log k)$-hypercontractivity for all $j \in [t]$. In other words, we used the property that for all $j \in [t]$, $u \in \reals^{|h^{-1}(j)|}$, 
  \begin{equation}
    \label{eq:prghci}
\ex\left[|\iprod{u}{X_{h^{-1}(j)}}|^{4\log k}\right] \leq (C \log k)^{2\log k}\,\|u\|^{4\log k}.    
  \end{equation}
Note that $X$ generated according to the generator $G$ satisfies both the above conditions: 1) For a fixed function $h$, the blocks are independent by definition and 2) the hypercontractivity inequality \eqref{eq:prghci} only involves the first $(4\log k)$-moments of the distribution of $X_{h^{-1}(j)}$. As a consequence, inequality \eqref{eq:prghci} holds for any $(4\log k)$-wise independent distribution over $\dpm^{|h^{-1}(j)|}$. 

We can now move from closeness in expectation to closeness in cdf distance by an argument similar to the proof of \tref{th:ipmain}, where we use \eref{eq:prgrp1} instead of \tref{th:ipsmooth}, to get 
\[ |\pr_{y \in_u \zo^r}\left[G(y) \in \ptp\right] - \pr_{Y \lfta \NN}\left[Y \in \ptp\right]| \leq \delta.\]
The theorem now follows from the above equation and \tref{th:ipmain}.
\end{proof}

\subsubsection{Approximate Counting for Integer Programs}
The PRG from \tref{th:prgmain} coupled with enumeration over all possible seeds immediately implies a quasi-polynomial time, deterministic algorithm for approximately counting, within a small additive error, the number of solutions to ``regular'' $\zo$-integer programs.
It turns out that ``regular'' integer programs correspond to a broad class of well-studied combinatorial problems.
For example, we obtain deterministic, approximate counting algorithms for {\em dense} set cover problems and $\{0,1\}$-contingency tables.  We obtain quasi-polynomial time algorithms even when there are a polynomial number of constraints (or polynomial number of rows in the contingency table setting).  As far as we know, there is no prior work giving nontrivial {\em deterministic} algorithms for counting solutions to integer programs with many constraints.

\newcommand{\ctrc}{\mathsf{CT}({\bf r},{\bf c}}

Here we discuss the case of {\em dense} set cover instances and remark that we get similar results for the special case of counting contingency tables. Covering integer programs are a fundamental class of integer programs and can be formulated as follows.
\begin{align}\label{eq:ipcovering}
\min &\sum_i X_i\nonumber\\
\text{s.t. }\sum_i a_{ij} X_i &\geq c_j,\; j = 1,\ldots,k,\\
X &\in \{0,1\}^n,\nonumber
\end{align}
where the coefficients of the constraints $a_{ij}$ and $c_{j}$ are all non-negative. An important special class of covering integer programs is set cover, which in turn is a generalization of many important problems in combinatorial optimization such as edge cover and multidimensional $\{0,1\}$-knapsack.

In the standard set cover problem, the input is a family of sets
$S_1,\ldots,S_n$ over a universe $U$ of size $k$ and an integer $t$.
The goal is to find a subfamily of sets $\calC$ such that $|\calC|
\leq t$ and the union of all the sets in $\calC$ equals $U$.  This
corresponds to a covering program (as given below) with $k$ constraints and
$n$ unknowns from $\{0,1\}$. 
\begin{align}\label{eq:setcovering}
\min &\sum_{i=1}^n X_i\nonumber\\
\text{s.t. }\sum_{i: j \in S_i} X_i &\geq 1,\; j \in U,\\
X &\in \{0,1\}^n,\nonumber
\end{align}
Call an instance of set cover $\epsilon$-dense if each element in $U$
appears in at least $1/\eps^2$ of the different sets $S_{i}$.
Clearly, all the linear constraints that appear in
\eref{eq:setcovering} are $\eps$-regular if the set cover instance is
$\epsilon$-dense. These constraints continue to be $\eps$-regular even
after translating from $\{0,1\}$ to $\dpm$ and appropriate
normalization. Thus, using the generator from \tref{th:prgmain} and
enumerating over all seeds to the generator, we have the following:

\vspace{-1ex}
\begin{theorem}
There exists a deterministic algorithm that, given instance of an
$\eps$-dense set covering problem with $k$ constraints over a universe
of size $n$, approximates the number of solutions to within an additive
error of at most $\delta 2^n$ in time $n^{\poly(\log k, 1/\delta)}$ as long as $\epsilon \leq \delta^{5}/(\log^{8.1} k)(\log (1/\delta))$.
\end{theorem}

We now elaborate on approximately counting the number of $\zo$ contingency tables. The problem of counting $\zo$-contingency tables is the following. Given, positive integers $n, k$ $n > k$, ${\bf r} = (r_1,\ldots,r_n) \in \mathbb{Z}^n$, ${\bf c} = (c_1,\ldots,c_k) \in \mathbb{Z}^k$ we wish to count the number of solutions, $\ctrc)$, to the following integer program whose solutions are matrices $X \in \zo^{ n \times k}$ with row and column sums given by ${\bf r}, {\bf c}$.
\begin{align*}
   \text{Find }X &\in \zo^{n \times k}\\
\text{s.t. }\sum_j X_{ij} &= r_i, 1 \leq i \leq n,\\
\sum_i X_{ij} &= c_j, 1\leq j \leq k.
\end{align*}

Observe that, after translating from $\zo$ to $\dpm$ and appropriately normalizing, solutions to the above integer program correspond to points from $\dpm^{n \times k}$ that lie in an intersection of $2(n+k)$-halfspaces each of which is $(1/\sqrt{k})$-regular (recall that the notion of regularity does not depend on the value of the $c_{i}$'s or $r_{j}$'s). 
Thus, as with dense instances of set cover, we can use \tref{th:prgmain} to count the number of $\zo$-contingency tables:
\begin{theorem}
There exists a deterministic algorithm that on input ${\bf r} \in \mathbb{Z}^n$, ${\bf c} \in \mathbb{Z}^k$, approximates $\ctrc)/2^{nk}$, the fraction of $\{0,1\}$-contingency tables with sums ${\bf r}, {\bf c}$, to within additive error $\delta$, and runs in time $n^{\poly(\log k, 1/\delta)}$.
\end{theorem}
We remark that using results of Wolff \cite{Wolff2007}, who shows hypercontractivity for various discrete distributions, we can approximately count number of solutions to dense set cover instances and contingency tables over most natural domains.

\ignore{
We can also extend the above algorithm to counting contingency tables over domains other than $\zo$ as well. Let $S \subseteq \mathbb{Z}$. The problem of $S$-contingency tables is the following. Given positive integers $n > k$, ${\bf r} = (r_1,\ldots,r_n), {\bf c} = (c_1,\ldots,c_k)$ estimate 
\[ \ctrc,S) = |\{X \in S^{n \times k} : \sum_j X_{ij} = r_i, i \in [n],\; \sum_i X_{ij} = c_j, j \in [k]\}|.\]

Let $\mu_S$ denote the uniform distribution over the set $S$. Observe that, we can use \tref{th:prgmain} to deterministically approximate $\ctrc,S)$ as in the $\zo$ case above, provided $\mu_S$ is both proper and hypercontractive. Moreover, it follows from the results of \cite{Wolff2007} that for many natural sets $S$ such as $\{0,1,2,\ldots, C\}$ for $C$ a constant, $\mu_S$ (after a translation, if needed) is both proper and hypercontractive. Thus, we can deterministically approximate the number of $S$-contingency tables, $\ctrc,S)$ in quasi-polynomial time for various discrete sets $S$.
}


\newcommand{\genn}{G_{\mathcal{N}}}

\subsection{Pseudorandom Generators for Polytopes in Gaussian Space}\label{sec:prgnormal}
We now prove \tref{th:prgnormal}.  We use an idea of Ailon and Chazelle \cite{AilonC2009} and the invariance of the Gaussian measure to unitary rotations to obtain PRGs with respect to $\NN$ for {\em all} polytopes.  Similar ideas were used by Meka and Zuckerman to obtain PRGs for spherical caps (i.e., the case of one hyperplane).   In our setting, we must prove that, with respect to a random rotation, {\em all} of the bounding hyperplanes become regular with high probability.  Such a tail bound requires applying hypercontractivity. 

Let $H \in \reals^{n \times n}$ be the normalized Hadamard matrix with $HH^T = I_n$ and $H_{ij} \in \{1/\sqrt{n},-1/\sqrt{n}\}$. Ailon and Chazelle show that for any $w \in \rn$, and a random diagonal matrix $D$ with uniformly random $\{1,-1\}$ entries, the vector $HDw$ is regular with high probability. We derandomize their observation using hypercontractivity. For a vector $x \in \rn$, let $D(x) \in \rnn$ be the diagonal matrix with diagonal entries $x$. 

\begin{lemma}\label{lm:prgn1}
There exists a constant $C> 0$ such that the following holds. For any $w \in \rn$, $\|w\| = 1$ , $0 < \delta < 1$ and any $(C\log(k/\delta))$-wise independent distribution $\calD$ over $\dpm^n$, \[\pr_{x \lfta \calD}\left[\,\|HD(x)w\|_4^4 \geq C\log^2 (k/\delta)/n\,\right] \leq \delta/k.\]
\end{lemma}
\begin{proof}
Fix a $w \in \rn$ and a $C\log(k/\delta)$-wise independent distribution $\calD$ for constant $C$ to be chosen later. Let random variable $Z = \|HD(x)w\|_4^4=\sum_i\left(\sum_l H_{il}x_lw_l\right)^4$ for $x \lfta \calD$. Note that $x$ satisfies $(2,q)$-hypercontractivity for $q \leq C\log(k/\delta)$. Now,
\begin{align*}
  \ex\left[Z^2\right] &= \sum_{i,j} \ex\left[\left(\sum_{l} H_{il} x_{l} w_l\right)^4\left(\sum_{l'} H_{jl'} x_{l'} w_{l'}\right)^4\right]\\
&\leq \sum_{ij} \sqrt{\ex\left[\left(\sum_l H_{il}x_lw_l\right)^8\right]\cdot \ex\left[\left(\sum_l H_{jl}x_lw_l\right)^8\right]}\\
&\text{\hspace{.5in}Cauchy-Schwarz inequality}\\
&\leq \sum_{i,j} 8^4 \,\left(\ex\left[\left(\sum_l H_{il} x_l w_l\right)^2\right]\right)^2\left(\ex\left[\left(\sum_l H_{jl} x_l w_l\right)^2\right]\right)^2\\
& \text{\hspace{.5in}$(2,8)$-hypercontractivity}\\
&= 8^4 \sum_{i,j} \frac{1}{n^4} =\frac{c}{n^2}.
\end{align*}
The last equality follows from the fact that $E\left[x_{i}x_{j}\right] = 0$ for $i \neq j$ and that each $H_{ij}^{2} = 1/n$. 
Observe that $Z$ is a degree $4$ multilinear polynomial over $x_1,\ldots,x_n$. Therefore, by $(2,q)$-hypercontractivity, \lref{lm:hypercon}, applied to the random variable $Z$, for $q \leq C \log(k/\delta)/4$,
\[ \ex\left[|Z|^q\right] \leq q^{2q} (\ex\left[Z^2\right])^{q/2} \leq \frac{c^{q/2}\,q^{2q}}{n^q}.\]
Hence, by Markov's inequality, for $\gamma > 0$,
\[ \pr\left[\,|Z| > \gamma\,\right] = \pr\left[\,|Z|^q > \gamma^q\,\right] \leq \left(\frac{c^{1/2}\,q^2}{\gamma n}\right)^q.\]
The lemma now follows by taking $q = 2 \log(k/\delta)$ and $\gamma = 2\,c^{1/2}\, q^2/n$.
\end{proof}

Let $G:\zo^r \rgta \dpm^n$ be the generator from \tref{th:prgmain} for $r = O((\log n \log k)/\epsilon)$. Let $G_1:\zo^{r_1} \rgta \dpm^n$ generate a $C \log(k/\delta)$-wise independent distribution, for constant $C$ as in \lref{lm:prgn1}. Generators $G_1$ as above with $r_1 = O(\log(k/\delta) \log n)$ are known. Define $\genn:\zo^{r_1}\times \zo^{r} \rgta \rn$ as follows:
\[ \genn(x,y) = D(G_1(x))H G(y).\]

We claim that $\genn$ $\delta$-fools all polytopes with respect to $\ugs^n$. 

\begin{proof}[of \tref{th:prgnormal}]
Recall that $\epsilon = \Omega(\delta^{5.1}/\log^{8.1}k) > 1/n^{.51}$. The seed length of $\genn$ is $r_1 + r = O(\log n \log k/\epsilon)$. Fix $W \in \rnn$. Observe that  $W^T \genn(x,y) = (HD(G_1(x))W)^T G(y)$. Now, from \lref{lm:prgn1} and a union bound it follows that
\begin{equation}\label{eq:prggns1}
 \pr_{x \in_u \zo^{r_1}}\left[\, \text{$HD(G_1(x))W$ is not $\epsilon$-regular}\,\right] \leq \delta.  
\end{equation}
Further, from the invariance of $\NN$ with respect to unitary rotations, for any $x \in \zo^{r_1}$, 
\[ \pr_{z \lfta \NN}\left[(HD(G_1(x))W)^T z \in \rect(\theta)\right] = \pr_{z \lfta \NN}\left[W^T z \in \rect(\theta)\right].\]
Thus, from \tref{th:prgmain} applied to $\ugs$, we get that for $HD(G_1(x))W$ $\epsilon$-regular, 
\begin{equation}\label{eq:prggns2}
\left|\pr_{y \in_u \zo^r}\left[\, (HD(G_1(x))W)^T G(y) \in \rect(\theta)\,\right] - \pr_{z \lfta \NN}\left[W^T z \in \rect(\theta)\right]\right| \leq \delta.  
\end{equation}

The theorem now follows from Equations \eqref{eq:prggns1}, \eqref{eq:prggns2}.
\end{proof}

\subsection{Pseudorandom Generators for Intersections of Spherical Caps}\label{sec:prgsphere}
\tref{th:prgspherical} follows from \tref{th:prgnormal} and the following new invariance principle for polytopes over $\spn$:  The proof uses Nazarov's bound on Gaussian surface area and the large deviation bound from \tref{th:gaussianld}.
\begin{lemma}\label{lm:ipspherical}
  For any polytope $\ptp$ with $k$ faces, 
\[ \left|\pr_{X \in_u \spn}\left[X \in \ptp\right] - \pr_{Y \lfta \NN}\left[Y/\sqrt{n} \in \ptp\right] \right| \leq \frac{C\log n\, \log k}{\sqrt{n}}.\]
\end{lemma}

\ignore{
\begin{lemma}\label{lm:ldrandvec}
For $Y \lfta \NN$,
\[ \pr\left[\,|\|Y\|-\sqrt{n}| > t\,\right] \leq a \exp(-b\, t^2),\]
where $a,b > 0$ are universal constants.
\end{lemma}}

\begin{proof}
Fix a polytope $\ptp(W,\theta)$. Let $X \in_u \spn$ and $Y \lfta \NN$. Note that $Y/\|Y\|$ is uniformly distributed over $\spn$. Fix $\delta = c/n^{1/2}$ for a constant $c$ to be chosen later. Observe that for $Y \lfta \NN$, and $u \in \reals^n$, $\|u\| = 1$, $\iprod{u}{Y}$ is distributed as $\ugs$. Hence, for any $u \in \reals^n, \|u\| = 1$, 
\[ \pr\left[\,|\iprod{u}{Y}| \geq \sqrt{\log(k/\delta)}\,\right] \leq \frac{\delta}{k}.\]
Therefore, by a union bound,
\[ \pr\left[\,\|W^T Y\|_\infty > \sqrt{\log(k/\delta)}\,\right] \leq \delta.\]
Further, by applying \tref{th:gaussianld} to the Euclidean norm (which has Lipschitz constant $1$) and the fact that $\ex\left[\|Y\|\right] = \Omega(\sqrt{n})$, we get
\begin{multline*}
 \pr\left[\, |W^T Y\|_\infty/\|Y\| > C\sqrt{\log (k/\delta)}/\sqrt{n}\,\right] \leq \\
\pr\left[\,\|W^T Y\|_\infty > \sqrt{\log(k/\delta)}\,\right] + \pr\left[\, \|Y\|_2 < \sqrt{n}/C\,\right] \leq \delta + 2\exp(-\Omega(n)) \leq 2 \delta,
\end{multline*}
for a sufficiently large constant $C$ and large $n$. Therefore, as $Y/\|Y\|$ is uniformly distributed over $\spn$, 
 
\[ \pr\left[\,\|W^T X\|_\infty > \sqrt{C \log(k/\delta)}/\sqrt{n}\,\right] \leq 2\delta,\]
From the above two equations, it follows that to prove the theorem we can assume that 
\begin{equation*}\label{eq:prgsp3}
 \|\theta\|_\infty < \sqrt{C \log(k/\delta)/n}.
\end{equation*} 
Now, applying \tref{th:gaussianld} (again to the Euclidean norm) and the above equation it follows that 
\begin{equation}\label{eq:prgsp2}
 \pr\left[\,|\|Y\| - \sqrt{n}|\,\|\theta\|_\infty \geq \sqrt{C\log(1/\delta)\log(k/\delta)/n}\,\right] \leq 2\delta.  
\end{equation}
Let $\lambda = \sqrt{C\log(1/\delta)\log(k/\delta)/n}$. Then, since $Y/\|Y\| \in_u \spn$
\begin{align*}\label{}
&\left|\,\pr\left[X \in \ptp\right] - \pr\left[Y/\sqrt{n} \in \ptp\right] \,\right| \\
&= \left|\,\pr\left[W^T X \in \rect(\theta)\right] - \pr\left[W^T Y/\sqrt{n} \in \rect(\theta)\right] \,\right|\\
&= \left|\,\pr\left[W^T Y \in \|Y\|\, \rect(\theta)\right] - \pr\left[W^T Y \in \sqrt{n}\, \rect(\theta)\right] \,\right|\\
&\leq \pr\left[\,|\|Y\| - \sqrt{n}|\,\|\theta\|_\infty \geq \lambda\,\right] \\
&+ \pr\left[\,W^TY \in \rect(\sqrt{n}\theta + \lambda \ok) \setminus \rect(\sqrt{n}\theta - \lambda \ok)\,\right] \\
&\leq  2\delta + O(\,\lambda \sqrt{\log k}\,).\text{\hspace{.5in}(\eref{eq:prgsp2}, \lref{lm:acrect})}
\end{align*}
The lemma now follows by choosing $\delta = c /n^{1/2}$ for a sufficiently large constant $c$. 
\ignore{
Fix $\delta = O(1/n^{1/4})$ to be chosen later. Note that we can assume without loss of generality that
\begin{equation}\label{eq:prgsp3}
 \|\theta\|_\infty < \sqrt{C \log(k/\delta)/n}, 
\end{equation}
for a sufficiently large constant $C$. For, if not $\pr_{X \in_u \spn}\left[X \in \ptp\right], \pr_{Y \lfta \NN}\left[Y/\sqrt{n} \in \ptp\right] = O(\delta) = O(1/n^{1/4})$. We now use the well known fact that $\|Y\|$ is concentrated about $\sqrt{n}$: for a sufficiently large constant $C > 0$,
\[ \pr\left[\,|\|Y\| - \sqrt{n}| \geq C \log^{1/2}(1/\delta)n^{1/4}\,\right] \leq \delta.\]
From \eref{eq:prgsp3} and the above equation it follows that 
\begin{equation}\label{eq:prgsp2}
 \pr\left[\,\|\,(\|Y\| - \sqrt{n})\theta\,\|_\infty \geq C \sqrt{\log(k/\delta)\log(1/\delta)}/n^{1/4}\,\right] \leq \delta.  
\end{equation}

From Equations \eqref{eq:prgsp1}, \eqref{eq:prgsp2} it follows that for $\lambda =  C(\log(k/\delta)\,\log(1/\delta))^{1/2}/n^{1/4}$, 
\begin{align*}
  \left|\pr_{X \in_u \spn}\left[X \in \ptp\right] - \pr_{Y \lfta \NN}\left[ W^T Y \in \rect(\sqrt{n} \theta)\right]\right| &\leq  \pr\left[\,\|\,(\|Y\| - \sqrt{n})\theta\,\|_\infty \geq \lambda\,\right] \\
&+ \pr_{Y \lfta \NN}\left[W^TY \in \rect(\sqrt{n}\theta + \lambda \ok) \setminus \rect(\sqrt{n}\theta - \lambda \ok)\right]\\
&\leq \delta + C \lambda \sqrt{\log k}.\text{\hspace{.5in}(\eref{eq:prgsp2}, \lref{lm:acrect})}
\end{align*}
The lemma now follows by setting $\delta = 1/n^{1/4}$. 
}
\end{proof}

\begin{proof}[of \tref{th:prgspherical}]
Define $\gsp:\zo^{r_1} \times \zo^r\rgta \spn$ by $\gsp(x,y) = G_\ugs(x,y)/\sqrt{n}$. It follows from \tref{th:prgnormal} and \lref{lm:ipspherical} that $\gsp$ fools polytopes over $\spn$ as in the theorem.
\end{proof}

\ignore{
\section{Misc proofs}
\begin{proof}[ of \clref{cl:ldtom}]
The claim follows from simple integration as follows. Without loss of generality suppose that $r$ is odd.
\begin{align*}
  \ex[|Z|^r] &= \int_{y= 0}^\infty \, pr\left[|Z|^r \geq y\right] dy\\
&\leq 2\int_{y=0}^{\infty} \exp(-2t^2\,y^{2/r}) dy \\
& = 2 \int_{z=0}^\infty \exp(-2 t^2 z^2) \,r z^{r-1} dz \text{\hspace{.5in} (substituting $z^r = y$)}\\
& = 2 r \left(\frac{1}{2t}\right)^{r-1} \frac{(r-1)!}{2^{(r-1)/2} ((r-1)/2)!} \text{\hspace{.5in} (the $(r-1)$'th moment of $\mathcal{N}(0,1/2t)$)}\\
&\leq \left(\frac{1}{2t}\right)^{r-1} r^{r/2}.
\end{align*}
\end{proof}
}

\section*{Acknowledgments}
Thanks to Fedja Nazarov for helping us compute an integral.  We had
useful conversations with Carly Klivans, Ryan O'Donnell, Alistair
Sinclair, Eric Vigoda, and David Zuckerman.  

\bibliographystyle{amsalpha}
\bibliography{poly}

\newcommand{\etalchar}[1]{$^{#1}$}
\providecommand{\bysame}{\leavevmode\hbox to3em{\hrulefill}\thinspace}
\providecommand{\MR}{\relax\ifhmode\unskip\space\fi MR }
\providecommand{\MRhref}[2]{%
  \href{http://www.ams.org/mathscinet-getitem?mr=#1}{#2}
}
\providecommand{\href}[2]{#2}
\begin{thebibliography}{KKMO07}

\bibitem[AC09]{AilonC2009}
Nir Ailon and Bernard Chazelle, \emph{The fast {J}ohnson--{L}indenstrauss
  transform and approximate nearest neighbors}, SIAM J. Computing \textbf{39}
  (2009), no.~1, 302--322, (Preliminary version in {\em 38th STOC}, 2006).

\bibitem[AM09]{AustrinM09}
Per Austrin and Elchanan Mossel, \emph{Approximation resistant predicates from
  pairwise independence}, Computational Complexity \textbf{18} (2009), no.~2,
  249--271.

\bibitem[Aus07]{Austrin2007}
Per Austrin, \emph{Balanced max 2-sat might not be the hardest}, Proc.\ $39$th
  ACM Symp.\ on Theory of Computing (STOC), ACM, 2007, pp.~189--197.

\bibitem[Ben90]{Bentkus1990}
Vidmantas~K. Bentkus, \emph{Smooth approximations of the norm and
  differentiable functions with bounded support in {B}anach space
  $l_\infty^k$}, Lithuanian Mathematical Journal \textbf{30} (1990), no.~3,
  223--230.

\bibitem[Ben03]{Bentkus2003}
\bysame, \emph{On the dependence of the {B}erry–-{E}sseen bound on
  dimension}, Journal of Statistical Planning and Inference \textbf{113}
  (2003), no.~2, 385--402.

\bibitem[BK10]{BansalK2010}
Nikhil Bansal and Subhash Khot, \emph{Inapproximability of hypergraph vertex
  cover and applications to scheduling problems}, Proc.\ $37$th International
  Colloquium of Automata, Languages and Programming (ICALP), Part I (Samson
  Abramsky, Cyril Gavoille, Claude Kirchner, Friedhelm {Meyer auf der Heide},
  and Paul~G. Spirakis, eds.), LNCS, vol. 6198, Springer, 2010, pp.~250--261.

\bibitem[BKS99]{BenjaminiKS1999}
Itai Benjamini, Gil Kalai, and Oded Schramm, \emph{Noise sensitivity of
  {B}oolean functions and applications to percolation}, Inst. Hautes \'{E}tudes
  Sci. Publ. Math. \textbf{90} (1999), no.~1, 5--43.

\bibitem[BO10]{BlaisO2010}
Eric Blais and Ryan O'Donnell, \emph{Lower bounds for testing function
  isomorphism}, Proc.\ $25$th IEEE Conference on Computational Complexity,
  IEEE, 2010, pp.~235--246.

\bibitem[BR07]{BeckRobins}
Matthias Beck and Sinai Robins, \emph{Computing the continuous discretely:
  Integer-point enumeration in polyhedra}, 1st ed., Undergraduate Texts in
  Mathematics, Springer, 2007.

\bibitem[BV08]{BarvinokV2008}
Alexander Barvinok and Ellen Veomett, \emph{The computational complexity of
  convex bodies}, Surveys on Discrete and Computational Geometry: Twenty Years
  Later (Jacob~E. Goodman, J\'{a}nos Pach, and Richard Pollack, eds.),
  Contemporary Mathematics, vol. 453, AMS, 2008, pp.~117--137.

\bibitem[CD03]{CryanD2003}
Mary Cryan and Martin~E. Dyer, \emph{A polynomial-time algorithm to
  approximately count contingency tables when the number of rows is constant},
  J. Computer and System Sciences \textbf{67} (2003), no.~2, 291--310.

\bibitem[Cha05]{Chatterjee2005}
Sourav Chatterjee, \emph{A simple invariance theorem}, 2005.

\bibitem[DFR08]{DinurFR2008}
Irit Dinur, Ehud Friedgut, and Oded Regev, \emph{Independent sets in graph
  powers are almost contained in juntas}, Geometric and Functional Analysis
  \textbf{18} (2008), no.~1, 77--97.

\bibitem[DHK{\etalchar{+}}10]{DiakonikolasHKMRST2010}
Ilias Diakonikolas, Prahladh Harsha, Adam Klivans, Raghu Meka, Prasad
  Raghavendra, Rocco Servedio, and Li-Yang Tan, \emph{Bounding the average
  sensitivity and noise sensitivity of polynomial threshold functions}, Proc.\
  $42$nd ACM Symp.\ on Theory of Computing (STOC), ACM, 2010, pp.~533--542.

\bibitem[DKN10]{DiakonikolasKN2010}
Ilias Diakonikolas, Daniel~M. Kane, and Jelani Nelson, \emph{Bounded
  independence fools degree-2 threshold functions}, Proc.\ $51$st IEEE Symp.\
  on Foundations of Comp.\ Science (FOCS), IEEE, 2010, pp.~11--20.

\bibitem[DMR09]{DinurMR2009}
Irit Dinur, Elchanan Mossel, and Oded Regev, \emph{Conditional hardness for
  approximate coloring}, SIAM J. Computing \textbf{39} (2009), no.~3, 843--873,
  (Preliminary version in {\em 38th STOC}, 2006).

\bibitem[dW08]{deWolf2008}
Ronald de~Wolf, \emph{A brief introduction to {F}ourier analysis on the
  {B}oolean cube}, Theory of Computing, Graduate Surveys \textbf{1} (2008),
  1--20.

\bibitem[Dye03]{Dyer2003}
Martin~E. Dyer, \emph{Approximate counting by dynamic programming}, Proc.\
  $35$th ACM Symp.\ on Theory of Computing (STOC), ACM, 2003, pp.~693--699.

\bibitem[Fel68]{Feller1}
William Feller, \emph{An introduction to probability theory and its
  applications, volume 1}, 3rd ed., Wiley, 1968.

\bibitem[Fel71]{Feller2}
\bysame, \emph{An introduction to probability theory and its applications,
  volume 2}, 2nd ed., Wiley, 1971.

\bibitem[FGRW09]{FeldmanGRW2009}
Vitaly Feldman, Venkatesan Guruswami, Prasad Raghavendra, and Yi~Wu,
  \emph{Agnostic learning of monomials by halfspaces is hard}, Proc.\ $50$th
  IEEE Symp.\ on Foundations of Comp.\ Science (FOCS), IEEE, 2009,
  pp.~385--394.

\bibitem[GKM10]{GopalanKM2010}
Parikshit Gopalan, Adam Klivans, and Raghu Meka, \emph{Polynomial-time
  approximation schemes for knapsack and related counting problems using
  branching programs}, 2010.

\bibitem[GOWZ10]{GopalanOWZ2010}
Parikshit Gopalan, Ryan O'Donnell, Yi~Wu, and David Zuckerman, \emph{Fooling
  functions of halfspaces under product distributions}, Proc.\ $25$th IEEE
  Conference on Computational Complexity, IEEE, 2010, pp.~223--234.

\bibitem[GW95]{GoemansW1995}
Michel~X. Goemans and David~P. Williamson, \emph{Improved approximation
  algorithms for maximum cut and satisfiability problems using semidefinite
  programming}, J. ACM \textbf{42} (1995), no.~6, 1115--1145, (Preliminary
  version in {\em 26th STOC}, 1994).

\bibitem[H{\aa}s01]{Hastad2001}
Johan H{\aa}stad, \emph{Some optimal inapproximability results}, J. ACM
  \textbf{48} (2001), no.~4, 798--859, (Preliminary Version in {\em 29th STOC},
  1997).

\bibitem[Hau92]{Haussler1992}
David Haussler, \emph{Decision theoretic generalizations of the {PAC} model for
  neural net and other learning applications}, Inf. Comput. \textbf{100}
  (1992), no.~1, 78--150, (Preliminary version in {\em 1st ALT}, 1990).

\bibitem[INW94]{ImpagliazzoNW1994}
Russell Impagliazzo, Noam Nisan, and Avi Wigderson, \emph{Pseudorandomness for
  network algorithms}, Proc.\ $26$th ACM Symp.\ on Theory of Computing (STOC),
  ACM, 1994, pp.~356--364.

\bibitem[Jan97]{janson1997}
S.~Janson, \emph{Gaussian hilbert spaces}, Cambridge Tracts in Mathematics,
  Cambridge University Press, 1997.

\bibitem[JS97]{JerrumS1997}
Mark Jerrum and Alistair Sinclair, \emph{The {M}arkov chain {M}onte {C}arlo
  method: An approach to approximate counting and integration}, Approximation
  Algorithms for {NP}-hard Problems (Dorit~S. Hochbaum, ed.), PWS Publishing
  Company, 1997.

\bibitem[Kal05]{Kalai2005}
Gil Kalai, \emph{Noise sensitivity and chaos in social choice theory}, Tech.
  Report 399, Center for Rationality and Interactive Decision Theory, Hebrew
  University of Jerusalem, 2005.

\bibitem[KKL88]{KahnKL1988}
Jeff Kahn, Gil Kalai, and Nathan Linial, \emph{The influence of variables on
  {B}oolean functions (extended abstract)}, Proc.\ $29$th IEEE Symp.\ on
  Foundations of Comp.\ Science (FOCS), IEEE, 1988, pp.~68--80.

\bibitem[KKMO07]{KhotKMO2007}
Subhash Khot, Guy Kindler, Elchanan Mossel, and Ryan O'Donnell, \emph{Optimal
  inapproximability results for {MAX-CUT} and other 2-variable {CSP}s?}, SIAM
  J. Computing \textbf{37} (2007), no.~1, 319--357, (Preliminary version in
  {\em 45th FOCS}, 2004).

\bibitem[KKMS08]{KalaiKMS2008}
Adam~Tauman Kalai, Adam~R. Klivans, Yishay Mansour, and Rocco~A. Servedio,
  \emph{Agnostically learning halfspaces}, SIAM J. Computing \textbf{37}
  (2008), no.~6, 1777--1805, (Preliminary version in {\em 46th FOCS}, 2005).

\bibitem[KOS04]{KlivansOS2004}
Adam~R. Klivans, Ryan O'Donnell, and Rocco~A. Servedio, \emph{Learning
  intersections and thresholds of halfspaces}, J. Computer and System Sciences
  \textbf{68} (2004), no.~4, 808--840, (Preliminary version in {\em 43rd FOCS},
  2002).

\bibitem[KOS08]{KlivansOS2008}
\bysame, \emph{Learning geometric concepts via {G}aussian surface area}, Proc.\
  $49$th IEEE Symp.\ on Foundations of Comp.\ Science (FOCS), IEEE, 2008,
  pp.~541--550.

\bibitem[KSS94]{KearnsSS1994}
Michael~J. Kearns, Robert~E. Schapire, and Linda Sellie, \emph{Toward efficient
  agnostic learning}, Machine Learning \textbf{17} (1994), no.~2--3, 115--141,
  (Preliminary version in {\em 5th COLT}, 1992).

\bibitem[LMN93]{LinialMN1993}
Nathan Linial, Yishay Mansour, and Noam Nisan, \emph{Constant depth circuits,
  {F}ourier transform, and learnability}, J. ACM \textbf{40} (1993), no.~3,
  607--620, (Preliminary version in {\em 30th FOCS}, 1989).

\bibitem[LT91]{LedouxT}
Michel Ledoux and Michel Talagrand, \emph{Probability in banach spaces:
  Isoperimetry and processes}, Springer, 1991.

\bibitem[Man94]{Mansour1994}
Yishay Mansour, \emph{Learning {B}oolean functions via the {F}ourier
  transform}, Theoretical Advances in Neural Computation and Learning (Vwani~P.
  Roychowdhury, Kai-Yeung Siu, and Alon Orlitsky, eds.), Kluwer Academic
  Publishers, 1994, pp.~391--424.

\bibitem[MH99]{MahajanH1999}
Sanjeev Mahajan and Ramesh Hariharan, \emph{Derandomizing approximation
  algorithms based on semidefinite programming}, SIAM J. Computing \textbf{28}
  (1999), no.~5, 1641--1663, (Preliminary version in {\em 36th FOCS}, 1995).

\bibitem[MOO05]{MosselOO2005}
Elchanan Mossel, Ryan O'Donnell, and Krzysztof Oleszkiewicz, \emph{Noise
  stability of functions with low influences invariance and optimality}, Proc.\
  $46$th IEEE Symp.\ on Foundations of Comp.\ Science (FOCS), IEEE, 2005,
  pp.~21--30.

\bibitem[Mos08]{Mossel2008}
Elchanan Mossel, \emph{Gaussian bounds for noise correlation of functions and
  tight analysis of long codes}, Proc.\ $49$th IEEE Symp.\ on Foundations of
  Comp.\ Science (FOCS), IEEE, 2008, pp.~156--165.

\bibitem[Mos11]{Mossel2011}
\bysame, \emph{A quantitative {A}rrow {T}heorem}, 2011.

\bibitem[MZ10]{MekaZ2010}
Raghu Meka and David Zuckerman, \emph{Pseudorandom generators for polynomial
  threshold functions}, Proc.\ $42$nd ACM Symp.\ on Theory of Computing (STOC),
  ACM, 2010, pp.~427--436.

\bibitem[Naz03]{Nazarov2003}
Fedor Nazarov, \emph{On the maximal perimeter of a convex set in {${\mathbb
  R}^n$} with respect to a {G}aussian measure}, Geometric Aspects of Functional
  Analysis (Israel Seminar 2001--2002), Lecture Notes in Mathematics, vol.
  1807/2003, Springer, 2003, pp.~169--187.

\bibitem[NN93]{NaorN1993}
Joseph Naor and Moni Naor, \emph{Small-bias probability spaces: Efficient
  constructions and applications}, SIAM J. Computing \textbf{22} (1993), no.~4,
  838--856, (Preliminary Version in {\em 22nd STOC}, 1990).

\bibitem[O'D04]{ODonnell2004}
Ryan O'Donnell, \emph{Hardness amplification within {NP}}, J. Computer and
  System Sciences \textbf{69} (2004), no.~1, 68--94, (Preliminary version in
  {\em 34th STOC}, 2002).

\bibitem[O'D08]{ODonnell2008}
\bysame, \emph{Some topics in analysis of {B}oolean functions}, Proc.\ $40$th
  ACM Symp.\ on Theory of Computing (STOC), ACM, 2008, pp.~569--578.

\bibitem[OW09]{ODonnellW2009}
Ryan O'Donnell and Yi~Wu, \emph{Conditional hardness for satisfiable 3-{CSP}s},
  Proc.\ $41$st ACM Symp.\ on Theory of Computing (STOC), ACM, 2009,
  pp.~493--502.

\bibitem[Per04]{Peres2004}
Yuval Peres, \emph{Noise stability of weighted majority}, 2004.

\bibitem[Pin94]{Pinelis1994}
Iosif Pinelis, \emph{Extremal probabilistic problems and hotelling’s {$T^2$}
  test under a symmetry condition}, Ann. Statist. \textbf{22} (1994), no.~1,
  357--368.

\bibitem[PR89]{PR-clt}
Vygantas Paulauskas and Alfredas Ra\v{c}kauskas, \emph{Approximation theory in
  the central limit theorem: Exact results in {B}anach spaces}, Kluwer Academic
  Publishers, 1989, (Translated from Russian).

\bibitem[Rag08]{Raghavendra2008}
Prasad Raghavendra, \emph{Optimal algorithms and inapproximability results for
  every {CSP}?}, Proc.\ $40$th ACM Symp.\ on Theory of Computing (STOC), ACM,
  2008, pp.~245--254.

\bibitem[Rot79]{Rotar1979}
Vladimir~Il'ich Rotar, \emph{Limit theorems for polylinear forms}, Journal of
  Multivariate Analysis \textbf{9} (1979), no.~4, 511 -- 530.

\bibitem[Shi00]{Shi2000}
Yaoyun Shi, \emph{Lower bounds of quantum black-box complexity and degree of
  approximating polynomials by influence of {B}oolean variables}, Inf. Process.
  Lett. \textbf{75} (2000), no.~1-2, 79--83.

\bibitem[Wol07]{Wolff2007}
Pawel Wolff, \emph{Hypercontractivity of simple random variables}, Studia Math
  \textbf{180} (2007), no.~3, 219--236.

\bibitem[Zie95]{Ziegler-polytope}
G\"{u}nter~M. Ziegler, \emph{Lectures on polytopes}, Graduate texts in
  Mathematics, vol. 152, Springer, 1995.

\end{thebibliography}


\end{document}